\documentclass{article} % For LaTeX2e
\usepackage{iclr2023_conference,times}
\usepackage[hidelinks]{hyperref}
\usepackage[toc,page,header]{appendix}
\usepackage{minitoc}
\usepackage{amsmath,amsfonts,bm}

\def\eqref#1{equation~\ref{#1}}
\def\1{\bm{1}}

\DeclareMathAlphabet{\mathsfit}{\encodingdefault}{\sfdefault}{m}{sl}
\SetMathAlphabet{\mathsfit}{bold}{\encodingdefault}{\sfdefault}{bx}{n}

\usepackage{algorithm}
\usepackage{algpseudocode}
\usepackage{hyperref}
\usepackage{url}
\usepackage{mathtools}
\usepackage{amsmath,amsthm,amssymb}
\usepackage{caption}
\usepackage{subcaption}

\usepackage[utf8]{inputenc} % allow utf-8 input
\usepackage[T1]{fontenc}    % use 8-bit T1 fonts
\usepackage{hyperref}       % hyperlinks
\usepackage{url}            % simple URL typesetting
\usepackage{booktabs}       % professional-quality tables
\usepackage{amsfonts}       % blackboard math symbols
\usepackage{nicefrac}       % compact symbols for 1/2, etc.
\usepackage{microtype}      % microtypography
\usepackage{xcolor}         % colors
\usepackage{wrapfig}

\newcommand{\Alphabet}{\mathcal{A}}
\newcommand{\M}{\mathcal{M}} % for motif
\newcommand{\X}{\mathbf{x}} % for structure (random variable)
\newcommand{\Xval}{x} % for structure (value it takes)
\newcommand{\Seqval}{s} % for sequence (value it takes
\newcommand{\Sc}{\mathcal{S}} % for scaffold

\newcommand{\scTM}{\texttt{scTM}}

\newcommand{\dt}[2]{#1^{(#2)}} % for diffusion step
\newcommand{\ptheta}{p_\theta} % for diffusion step

\newcommand{\replace}{\mathrm{Repl}}
\newcommand{\modelname}{\texttt{ProtDiff}}
\newcommand{\protmpnn}{\texttt{ProteinMPNN}}
\newcommand{\samplingname}{\texttt{SMCDiff}}

\newcommand{\pall}[2]{p_{#1}(#2)}
\newcommand{\preplace}[2]{p^\replace_{#1}(#2)}

\newcommand{\preplacescaf}[2]{p^\replace_{\Sc,#1}(#2)}
\newcommand{\pmotif}[2]{p_{\M,#1}(#2)}
\newcommand{\pscaf}[2]{p_{\Sc,#1}(#2)}
\newcommand{\qall}[2]{q_{#1}(#2)}
\newcommand{\qmotif}[2]{q_{\M,#1}(#2)}
\newcommand{\qscaf}[2]{q_{\Sc,#1}(#2)}
\newcommand{\KL}[2]{\mathrm{KL}\left[#1 \| #2\right]}
\newcommand{\EKL}[2]{\mathrm{EKL}\left[#1 \| #2\right]}

\usepackage [autostyle, english = american]{csquotes}
\MakeOuterQuote{"}

\newcommand{\R}{\mathbb{R}}

\newcommand{\pr}{\mathbb{P}}

\newcommand{\Var}{\mathrm{Var}}
\newcommand{\inv}{^{-1}}

\usepackage[capitalize]{cleveref}

\def\[#1\]{\begin{align}#1\end{align}}        % numbered
\def\(#1\){\begin{align*}#1\end{align*}}     % unnumbered

\newcommand{\revision}[1]{#1}

\newtheorem{theorem}{Theorem}[section]
\newtheorem{corollary}[theorem]{Corollary}
\newtheorem{prop}[theorem]{Proposition}

\newtheorem{lemma}{Lemma}[section]
\theoremstyle{remark}

\title{Diffusion probabilistic modeling of protein backbones in 3D for the motif-scaffolding problem}

\author{%
  Brian L. Trippe\thanks{Contributed equally to this work.}$^{\ \ }$\thanks{Massachusetts Institute of Technology}\\
  \texttt{btrippe@mit.edu} \\
 \And
    Jason Yim$^{*\dagger}$ \\
  \texttt{jyim@mit.edu} \\
 \And
    Doug Tischer \thanks{University of Washington}\\
  \texttt{dtischer@uw.edu}\\
 \And
  David Baker$^{\ddagger}$\\
  \texttt{dabaker@uw.edu} \\
 \And
  Tamara Broderick$^{\dagger}$\\
  \texttt{tbroderick@mit.edu} \\
 \And
  Regina Barzilay$^{\dagger}$\\
  \texttt{regina@csail.mit.edu} \\
 \And
  Tommi Jaakkola$^{\dagger}$\\
  \texttt{tommi@csail.mit.edu} \\
}

\iclrfinalcopy % Uncomment for camera-ready version, but NOT for submission.
\begin{document}

\maketitle

\doparttoc % Tell to minitoc to generate a toc for the parts
\faketableofcontents % Run a fake tableofcontents command for the partocs

%\part{} % Start the document part
%\parttoc % Insert the document TOC

%\vspace{-45pt}

\begin{abstract}
% Version on 2022-05-16
Construction of a \emph{scaffold} structure that supports a desired \emph{motif}, conferring protein function, shows promise for the design of vaccines and enzymes.
But a general solution to this motif-scaffolding problem remains open.
Current machine-learning techniques for scaffold design are either limited to unrealistically small scaffolds (up to length 20) or struggle to produce multiple diverse scaffolds.
We propose to learn a distribution over diverse and longer protein backbone structures via an E(3)-equivariant graph neural network.
We develop \samplingname{} to efficiently sample scaffolds from this distribution conditioned on a given motif; our algorithm is the first to theoretically guarantee conditional samples from a diffusion model in the large-compute limit.
We evaluate our designed backbones by how well they align with AlphaFold2-predicted structures.
We show that our method can (1) sample scaffolds up to 80 residues and (2) achieve structurally diverse scaffolds for a fixed motif.\footnote{Code: \url{https://github.com/blt2114/ProtDiff_SMCDiff}}
\looseness=-1

\end{abstract}
\setcounter{footnote}{0} 

\section{Introduction}\label{sec:intro}
A central task in protein design is creation of a stable \textit{scaffold} to support a target \textit{motif}. 
\revision{Here, motifs are structural protein fragments imparting biological function while scaffolds stabilize the motif's structure.}
Vaccines and enzymes have already been designed by solving certain instances of this \emph{motif-scaffolding} problem \citep{procko2014computationally, correia2014proof, jiang2008novo, siegel2010computational}.
%\revision{\emph{Motifs} are structural protein fragments imparting biological function while \emph{scaffolds} are additional protein structures surrounding the motif to stabilize the motif's shape and folding. For example, the Fab domain of an antibody can be thought of as having two parts: strands in a $\beta$-sandwich (the scaffold) which support the binding loops (the motif).}
%A central task in protein design is creation of a stable scaffold to support a target motif.
%Vaccines and enzymes have already been designed by solving certain instances of this \emph{motif-scaffolding} problem 
However, successful solutions to this problem in the past have necessitated substantial expert involvement and laborious trial and error. Machine learning (ML) offers the hope to automate, and better direct this search. 
But existing ML approaches face one of two major roadblocks.
\revision{First, many methods} do not build scaffolds longer than  about 20 residues.
For many motif sizes of interest, the resulting proteins would be smaller than the shortest commonly-studied simple protein folds (35--40 residues) \citep{gelman2014fast}.
\revision{Second, while other methods may generate longer scaffolds using stochastic search techniques, they} require hours of computation to generate a single plausible scaffold \citep{wang2021deep,anishchenko2021novo,tischer2020design}. Moreover, when a plausible scaffold is found, it remains to be experimentally validated. Therefore, it is desirable to return not just a single scaffold but rather a set of scaffolds exhibiting diverse sequences and structural variation to increase the likelihood of success in practice.

In the present work, we demonstrate the promise of a particular generative modeling approach within ML for efficiently returning a diverse set of motif-supporting scaffolds.
Generative models have been shown to capture a distribution over diverse protein structures \citep{lin2021deep}. But it is not clear how to handle conditioning (on the motif) using these approaches. Diffusion probabilistic models (DPMs) offer a potential alternative; not only do they provide a more straightforward path to handling conditioning, but they have also enjoyed success generating small-molecules in 3D \citep{hoogeboom2022equivariant}. Extending DPMs to protein structures, though, is non-trivial; since proteins are larger than small molecules, modeling proteins requires handling the sequential ordering of residues and long-range interactions. Finally, while existing models often generate distance matrices \citep{anand2018generative,lin2021deep}, we instead focus on generating a full set of 3D coordinates, which should improve designability in practice. Our resulting model, \modelname{}, is similar to concurrent work on E(3)-equivariant diffusion models for molecules \citep{hoogeboom2022equivariant}, but with modifications specific to protein structure. 
Moreover, we develop a novel motif-scaffolding procedure based on Sequential Monte Carlo, \samplingname{}, that repurposes an unconditionally trained DPM for conditional sampling.
% In our case, we condition on the motif structure, a task analogous to \emph{inpainting} \citep{saharia2022palette}.
\revision{We prove that if a DPM matches the data distribution, \samplingname{} is guaranteed to provide \emph{exact} conditional samples in a large-compute limit;
this property contrasts with} previous methods \citep{song2020score, zhou20213d}, which we show introduce non-trivial approximation error that impedes performance.
Our final motif-scaffolding generative framework, then, has two steps (\cref{fig:schematic}): first we train \modelname{} to learn a distribution over protein backbones, and then we use \samplingname{} with \modelname{} to inpaint arbitrary motifs.

Ours is the first machine-learning method to construct scaffolds
%Our results demonstrate the first instance of machine learning methods constructing motif-scaffolds
longer than 20 residues around motifs --- we build up to 80 residues scaffolds on a test case.
Beyond our progress on the motif-scaffolding problem, we provide the following technical contributions:
(1) we introduce a protein-backbone generative model in 3D -- with the ability to generate backbone samples that structurally agree with \texttt{AlphaFold2} predictions, and
(2) we develop a novel conditional sampling algorithm for inpainting.

\begin{figure}
\centering
\includegraphics[width=\textwidth]{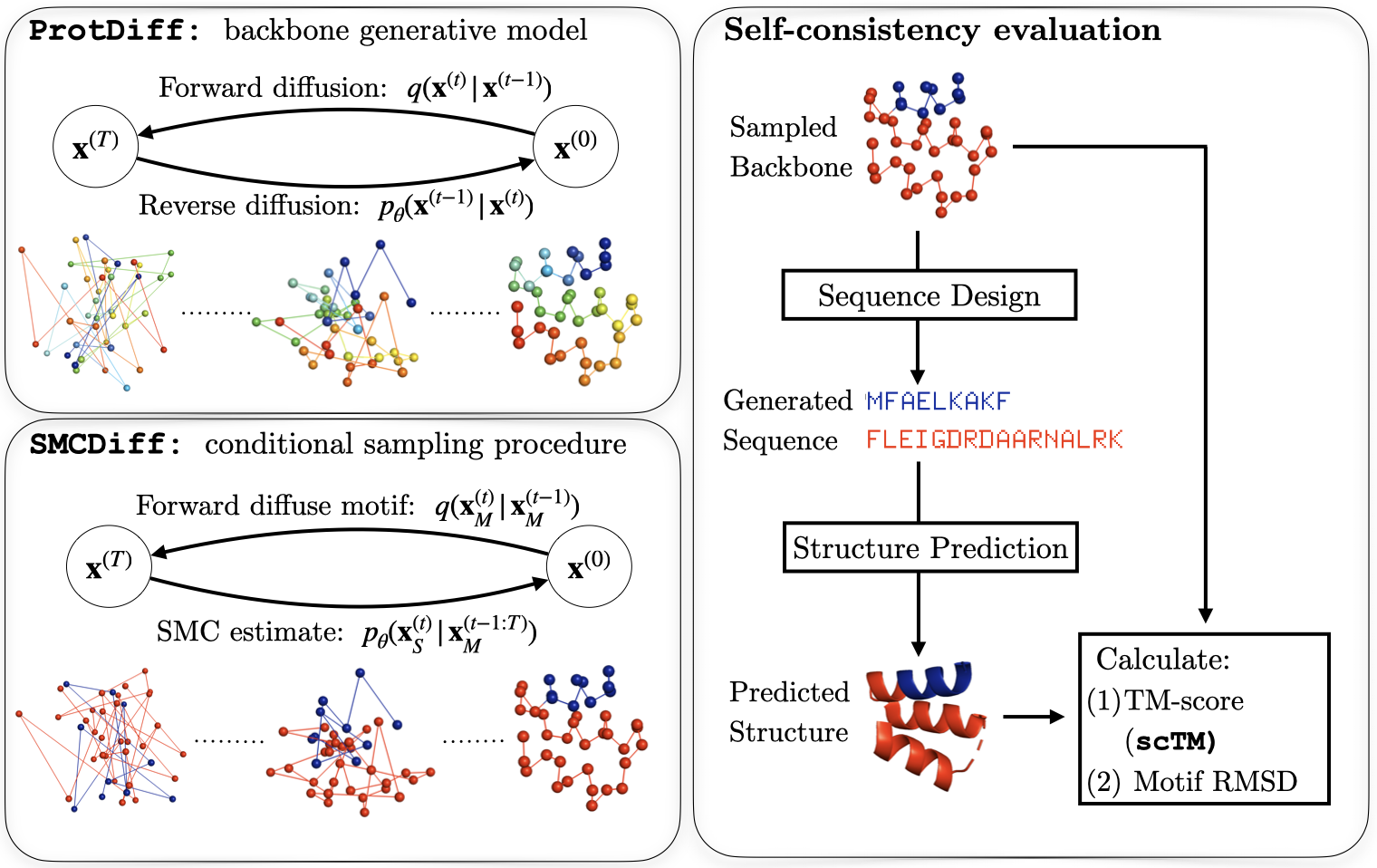}
\caption{
Overview of the conditional generative modeling approach to the motif-scaffolding problem. We train our new protein backbone diffusion model, \modelname, to generate realistic protein backbone structures. Next, we run \samplingname, our conditional sampling algorithm, with \modelname{} to generate scaffolds (colored in red) conditioned on the motif (colored in blue). 
% Training: a protein backbone diffusion model (\modelname) is trained to generate realistic protein backbone structures. Sampling: once trained, we use a conditional sampling algorithm (\samplingname) to generate scaffolds (colored in red) conditioned on the motif (colored in blue).
For self-consistency evaluation, we use a pretrained fixed-backbone sequence-design model (\protmpnn{} \citep{dauparas2022protmpnn}) to generate the scaffold sequence from a sampled backbone.
We then input the sequence to a structure prediction model, in our case \texttt{AlphaFold2} (AF2) \citep{Jumper2021HighlyAP}, to generate the full protein structure from the generated sequence.
We compare the backbone of the predicted structure with the original backbone structure using TM-score \citep{xu2010significant} and root-mean-square-distance (RMSD) for the motif.
}\label{fig:schematic}
\end{figure}

\subsection{Related work}\label{sec:related}

\textbf{Motif-scaffolding}. Past approaches have sought to scaffold a motif with native or prespecified protein fragments, but are limited to finding a suitable match in the Protein Data Bank (PDB) and cannot adapt the scaffold to compensate for slight structural mismatches \citep{cao2021robust,silva2016motif,yang2021bottom,sesterhenn2020novo,linsky2020novo}.
More recently \citet{wang2021deep} used pre-trained protein structure prediction networks to recapitulate native scaffolds, but this method failed to generate scaffolds longer than 20 residues and can output only a single candidate scaffold rather than a diverse set.
By contrast, our goal is to sample diverse, long scaffolds.

\textbf{Diffusion models for molecule generation}.
Several concurrent works have extended equivariant diffusion models to molecule generation.
\cite{anand2022protein} extended diffusion models for generation of protein backbone frames and sequences \textit{conditioned} on secondary structure adjacency matrices. Similarly, \cite{luo2022antigen} focused on CDR-loop generation using diffusion models conditioned on non-CDR regions of the antibody-antigen. Our method does not require conditioning and is applicable to general proteins.
\cite{lee2022proteinsgm} approach the same problem as our work but build diffusion models over 2D distances matrices that requires post-processing to produce 3D structures through Rosetta minimization. We demonstrate capability of diffusion models to directly model 3D coordinates of proteins.
\citet{hoogeboom2022equivariant} developed an equivariant diffusion model (EDM) for generating small molecules in 3D.
However, because EDM does not enforce a spatial ordering of the atoms that compose small molecules, as we describe in \Cref{sec:experiments}, it does not learn a coherent chain structure as needed in proteins.

% \textbf{3D molecule generation}. \citet{hoogeboom2022equivariant} concurrently developed an equivariant diffusion model (EDM) for generating molecules in 3D. EDM generates both 3D coordinates and atomic types.
% However, because EDM does not enforce a spatial ordering of the atoms that compose small molecules, as we describe in \Cref{sec:experiments}, it does not learn a coherent chain structure as needed in proteins.

%RB instead of "fail" say "has challenges" more neutral, the same meaning
%As described in \cref{sec:method}, \modelname{} is closely related to EDM. Instead of jointly modeling sequence and backbone, we make use of the recent fixed backbone sequence design methods and simplify our problem to only generating the 3D coordinates.
\looseness=-1

\textbf{Inpainting and conditional sampling in diffusion models}. Point-Voxel Diffusion (PVD) \citep{zhou20213d} is a 3D diffusion model for generating shapes from the ShapeNet dataset. Though trained to generate shapes unconditionally, 
PVD completes (or \emph{inpaints}) full shapes when a partial point cloud is fixed during inference.
For general diffusion models, \citet{song2020score} proposed an alternative inpainting approach and remarked that this approach produces \emph{approximate} conditional samples.
However, these methods do not provide theoretical guarantees,
and when we compare them to \samplingname{}, we find that their approximation error impedes performance when applied to motif-scaffolding.
%We compare \samplingname{} to both methods in \Cref{sec:conditional_sampling} - showing \samplingname{} has better theoretical justification and empirical results for motif-scaffolding.
\citet{saharia2022palette} developed an inpainting diffusion model by training a diffusion model to denoise randomly generated masked regions while unmasked regions were unperturbed. However, their approach requires a detailed data augmentation strategy that does not exist for proteins.
% Another class of works involve conditional diffusion which attempt to control the diffusion process with either latents or external scores from a classifier (\citet{ilvr, ramesh2022hierarchical, sinha2021d2c}). The assumptions made by these methods, all on images, do not directly transfer to proteins.

We describe additional related work on protein generative models in \Cref{sec:additional_related}.

\section{Preliminaries}\label{sec:preliminaries}

\subsection{The motif-scaffolding problem}\label{sec:motif_scaffolding_formulation}
A protein can be represented by its amino acid sequence and backbone structure.
Let $\Alphabet$ be the set of 20 genetically-encoded amino acids.
We denote the sequence of an $N$-residue protein by $\Seqval \in \Alphabet^N$ and its C-$\alpha$ backbone coordinates in 3D by $\X = [\X_1,\dots,\X_N]^T \in \R^{N,3}$.
%\revision{Generally speaking, a protein's ground state conformation} 
\revision{We describe a protein as having a fixed structure} that is a function of its sequence,
so we may write $\X(\Seqval)$. 
We divide the $N$ residues into the functional motif $\M$ and the scaffold $\Sc,$ such that $\M \cup \Sc = \{1, 2,\dots, N\}.$
%It will be convenient to write $\X = ( \X_\M, \X_\Sc)$.
The goal is to identify, given the motif structure $\X_\M$, sequences $\Seqval$ 
whose structure recapitulates the motif to high precision $\X(\Seqval)_\M \approx \X_\M$.
%In particular, we will say that $S \epsSats (S_\M, X_\M)$ if $\RMSD\left(X_M, \structure{S}[k:k+M]\right) <\epsilon$,
% here $\RMSD$ is defined for the best alignment (e.g.\ as computed via the Kabsch algorithm ).\footnote{\BLT{Add a textbook reference for Kabsch algorithm.}}
\Cref{sec:supp_caveats} discusses several caveats of this simplified framing \revision{(e.g.\ our assumption of static structures)}.
%The fundamental challenge is that there is an exponentially large number of possible scaffold sequences and we can't test them all.

%\textbf{Generative modeling framing of motif scaffolding}.
%%In light of recent solutions to the fixed backbone sequence design problem,
%%we focus on finding compatible scaffold backbones, $X_\Sc$, and take a generative modeling approach.
%We let $q(\cdot)$ denote a distribution supported on designable backbone that can be segmented into a motif and scaffold $\Xval=(\Xval_\M, \Xval_\Sc)$.
%% For example, $q(\X)$ could be the backbones of all naturally occurring proteins.
%If a motif of interest is within the support of $q(\Xval_\M),$ then the conditional $q(\Xval_\Sc\mid x_\M)$ is supported on designable scaffolds that can complete the motif.
%A machine learning approach to the motif scaffolding problem can be divided into two steps as shown in \cref{fig:schematic}.
%In the first step, we fit a generative model $p_\theta$ to $q$. Once trained, we may sample possible scaffolds from $p_\theta(\Xval_\Sc\mid \Xval_\M)$ and use existing backbone sequence methods to sample the scaffold sequence $\Seqval_\Sc$ which can then be ran through a protein folding model to generate the full atomic protein structure.

\subsection{Diffusion probabilistic models}
\label{sec:dpm}
Our approach to the motif-scaffolding problem builds on denoising diffusion probabilistic models (DPMs) \citep{sohl2015deep}.
We follow the conventions and notation set by \citet{ho2020denoising}, which we review here.
DPMs are a class of generative models based on a reversible, discrete-time diffusion process.
The \emph{forward process} starts with a sample $\dt{\X}{0}$ from an unknown data distribution $q$, with density denoted by $q(\dt{\X}{0}),$
and iteratively adds noise at each step $t$.
By the last step, $T$, the distribution of $\dt{\X}{T}$ is indistinguishable from an isotropic Gaussian: $\dt{\X}{T} \sim \mathcal{N}(\dt{\X}{T}; 0, \mathbf{I}).$
Specifically, we choose a variance schedule $\dt{\beta}{1}, \dt{\beta}{2},\dots, \dt{\beta}{T}$, and define the transition distribution at step $t$ as $q(\dt{\X}{t} \mid \dt{\X}{t-1}) = \mathcal{N}(\dt{\X}{t};\sqrt{1 - \dt{\beta}{t}}\dt{\X}{t-1}, \dt{\beta}{t} \mathbf{I})$.

DPMs approximate $q$ with a second distribution $\ptheta$ by learning the transition distribution of the \emph{reverse process} at each $t,$ $\ptheta(\dt{\X}{t-1} \mid \dt{\X}{t}).$
We follow the conventions set by \citet{ho2020denoising} in our parameterization and choice of objective.
In particular, we take $\ptheta(\dt{\X}{t-1}\mid \dt{\X}{t})= \mathcal{N}(\dt{\X}{t-1}; \mu_\theta(\dt{\X}{t}, t), \dt{\beta}{t}\mathbf{I})$ 
with
$\mu_\theta(\dt{\X}{t}, t)=\frac{1}{\sqrt{\dt{\alpha}{t}}}\left(\dt{\X}{t} - \frac{\dt{\beta}{t}}{\sqrt{1-\dt{\bar{\alpha}}{t}}}\epsilon_\theta(\dt{\X}{t},t)\right),$
$\dt{\alpha}{t}:=1-\dt{\beta}{t}$, and $\dt{\bar\alpha}{t} := \prod_{s=1}^t \dt{\alpha}{t}.$
We implement $\epsilon_\theta(\dt{\X}{t},t)$ as a neural network.
For training, we marginally sample $\dt{\X}{t}\sim q(\dt{\X}{t}\mid \dt{\X}{0})$ from the forward process as $\dt{\X}{t}=\sqrt{\dt{\bar \alpha}{t}}\dt{\X}{0} + \sqrt{1-\dt{\bar \alpha}{t}}\epsilon$ and minimize the  objective 
$T^{-1} \sum_{t=1}^T \mathbb{E}_{q(\dt{\X}{0},\dt{\X}{t})} \left[\|\epsilon - \epsilon_\theta(\dt{\X}{t}, t)\|^2\right]$ by stochastic optimization \citep[Algorithm 1]{ho2020denoising}.
To generate samples from $\ptheta(\dt{\X}{0}),$  we simulate the reverse process.
That is, we sample noise for time $T$ as $\dt{\X}{T} \sim \mathcal{N}(0, \mathbf{I}),$ 
and then for each $t=T-1,\dots,0,$  we simulate progressively ``de-noised'' samples as $\dt{\X}{t}\sim \ptheta(\dt{\X}{t}\mid \dt{\X}{t+1}).$

\section{\modelname: A diffusion model of protein backbones in 3D}\label{sec:method}
Implementation of diffusion probabilistic models requires choosing an architecture for the neural network $\epsilon_\theta(\X^{(t)}, t)$ introduced abstractly in \Cref{sec:dpm}.
In this section we describe \modelname, which corresponds to the choice of $\epsilon_\theta(\X^{(t)}, t)$ as a translation and rotation equivariant graph neural network tailored to modeling protein backbones.
\revision{We leave architectural and input encoding details to \cref{sec:supp_protdiff}.}

\textbf{The challenge of modeling points in 3D.}
The properties and functions of proteins are dictated by the relative geometry of their residues, and are invariant to the coordinate system chosen to encode them.
Recent work on neural network modeling of 3D data has found, both theoretically and empirically, that neural networks constrained to satisfy geometric invariances can provide inductive biases that improve generalization and training efficiency \citep{batzner2022egnn}.
Motivated by this observation, we parameterize $\epsilon_\theta$ by an equivariant graph neural network (EGNN) \citep{satorras2021en}, which in 3D is equivariant to transformations in the Euclidean group.
\citet{xu2022geodiff} proved that if $\epsilon_\theta$ is equivariant to a group then $\ptheta$ is invariant to the same group.
% Furthermore, \citet{xu2022geodiff} observed that one may impart geometric invariances to a DPM $\ptheta$ by constructing $\epsilon_\theta$ to the corresponding equivariances.
% Motivated by this observation, we parameterize $\epsilon_\theta$ by an equivariant graph neural network (EGNN)\citep{satorras2021en}, which in 3D is equivariant to transformations in the Euclidean group.

%A consequence of EGNN that is equivariant to rotations, translations and reflections.
%by incorporating sequence distance features and 
%that we tailor to work on protein backbone structures.

\textbf{Tailoring EGNN to protein backbones.}
We now describe our EGNN implementation, which we tailor to protein backbones and DPMs through the choice of edge and node features.
%Proteins backbones comprise of a long, linear chain and have structures strongly determined by interactions between residues distant from one another in the chain.
To model every pairwise residue interaction, we represent backbones by a fully connected graph.
% Because non-local interactions between residues are integral to protein structure, we represent protein backbones by a densely connected graph.
%To accommodate all possible non-local interactions in a graph neural network, we represent protein backbones by a \emph{densely connected} graph.
Each node in the graph is indexed by $n=1,\dots, N,$ and corresponds to a residue.
We associate each node with coordinates $\X_n \in \R^3$ and $D$ features $h_n \in\R^D.$
For each pair of nodes $n, n^\prime$ we define an edge and associate it with edge features. % $a_{nn^\prime}\in \R^D.$
We construct our EGNN by stacking $L$ equivariant graph convolutional layers (EGCL)  %\revision{previously described in \cite{satorras2021en} (see \cref{sec:supp_protdiff} for details)}
. 
Each layer takes node coordinates and features as input, and outputs updated coordinates and features with \revision{the first layer taking initial values $(\X, h)$.}
% We take the input to the first layer to be the initial coordinates and features $(\X^0, h^0)=(\X, h).$
% Each layer $l=1, \dots, L$ defines an update as $(\X^{l}, h^{l})=\text{EGCL}[\X^{l-1}, h^{l-1}]$ where for each node $n$
% \begin{gather*}
%     \label{eq:egnn}
%     \X_n^{l} = \X_n^{l-1} + \sum_{{n^\prime} \ne n} \vec \omega_{nn^\prime}\cdot\phi_\X(h_n^{l-1}, h_{n^\prime}^{l-1}, d_{n{n^\prime}}, a_{n{n^\prime}}) \quad \text{and}\quad
%     h_n^{l} = \phi_h(h_n^{l-1}, m_n), \quad \text{for}\\
%     \vec \omega_{nn^\prime} = \frac{\X_n^{l-1} - \X_{n^\prime}^{l-1}}{\sqrt{d_{nn^\prime}} + \gamma}, \quad 
%     m_n = \sum_{n^\prime \neq n}\phi_e(h_n^{l-1}, h_{n^\prime}^{l-1}, d_{n{n^\prime}}, a_{nn^\prime}),\text{ and } d_{nn^\prime} = \|\X_n^{l-1} - \X_{n^\prime}^{l-1}\|_2^2.
% \end{gather*}
% $\phi_e, \phi_h$, and  $\phi_\X$ are fully connected neural networks,
% and $\gamma$ is a small positive constant included for numerical stability.
We write the output of EGNN after $L$ layers as $\hat{\X} = \text{EGNN}[\X, h]$.
In the context of diffusion models, we predict the noise at time $t$ with the following parameterization:
\begin{equation}
    \epsilon_\theta(\dt{\X}{t}, t) = \hat{\X} - \dt{\X}{t}, \quad \hat{\X} = \text{EGNN}[\dt{\X}{t}, h(t)].
      %\\
    %\dt{\Xval}{t}_{\text{scaled}} = \frac{\sqrt{\dt{\bar \alpha}{t}}\dt{\Xval}{0}}{10} + \sqrt{1-\dt{\bar \alpha}{t}}\epsilon \label{eq:scaling}
    \label{eq:protdiff}
\end{equation}
%\revision{Input protein coordinates are represented in nanometers and translated to have center of mass at the origin.}

We now describe our choice of node and edge features.
Our choice is motivated by the linear chain structure of protein backbones;
residues close in sequence are necessarily close in 3D space.
To allow this chain constraint to be learned more easily, we fix an ordering of nodes in the graph to correspond to sequence order.
We include as edge features positional offsets as done in \cite{ingraham2019generative}, which we represent using sinusoidal positional encoding features \citep{vaswani2017transformer}.
% as
% \begin{align*}
%     a_{nn^\prime} = \begin{bmatrix}\varphi(n-n^\prime, 1)\\
%     \vdots \\
%     \varphi(n-n^\prime, D)
%     \end{bmatrix}, \text{ where }\
%     \varphi(x, k) = \begin{cases}
%     \text{sin}\left(x\cdot \pi/ N^{2\cdot k/D}\right) ,& k\mod 2 = 0 \\
%     \text{cos}\left(x\cdot \pi/ N^{2\cdot (k-1)/D}\right) ,& k\mod 2 = 1.
%     \end{cases}
%     % \label{eq:pos_feats}
% \end{align*}
For node features, we similarly use a sinusoidal encoding of sequence position as well as of the diffusion time step $t$ following \citet{kingma2021vdm}.
We additionally process the time encoding to be orthogonal to the positional encoding
%\revision{(see \cref{sec:supp_protdiff} for details)}
.

\section{\samplingname:  Conditional sampling in diffusion models by particle filtering}\label{sec:conditional_sampling}
The second stage of our generative modeling approach to the motif-scaffolding problem is to sample scaffolds $\dt{\X_\Sc}{0}$ from 
$\ptheta(\dt{\X_\Sc}{0} \mid \dt{\X_\M}{0})$.
\Cref{sec:replacement} discusses the intractability of sampling from $\ptheta(\dt{\X_\Sc}{0}\mid \dt{\X_\M}{0})$ exactly
and the limitations of a simple approximation introduced by \citep{song2020score}.
In \Cref{sec:smc}, we then frame computation of $\ptheta(\dt{\X_\Sc}{0} | \dt{\X_\M}{0})$ as a sequential Monte Carlo (SMC) problem \citep{doucet2001sequential}
and approximate it with a particle filtering algorithm (\Cref{alg:particle}).

\subsection{The challenge of conditional sampling and the error of the replacement method}\label{sec:replacement}
The conditional distributions of a DPM are defined implicitly through the steps of the reverse process.
We may write the conditional density explicitly as
\(
\ptheta(\dt{\X_\Sc}{0}\mid \dt{\X_M}{0})\propto \ptheta(\dt{\X_\Sc}{0}, \dt{\X_M}{0})
&= \ptheta(\dt{\X}{0}) = \int \ptheta(\dt{\X}{T}) \prod_{t=0}^{T-1}\ptheta(\dt{\X}{t}\mid \dt{\X}{t+1}) d\dt{\X}{1:T}.
\)
However, the high-dimensional integral on the right-hand side above is intractable (both analytically and numerically) to compute.

\begin{wrapfigure}{r}{0.5\textwidth}
\vspace{-27pt}
% \begin{figure}[!t]
\begin{minipage}{0.50\textwidth}
\begin{algorithm}[H]
\centering
\caption{
    \samplingname: Particle filtering for conditionally sampling from unconditional diffusion models
}
\label{alg:particle}
\begin{algorithmic}[1]
% %\LinesNumbered
    \State{\textbf{Input:} $\dt{\X_\M}{0}$ (motif), $K$ (\# particles)}
    \State{\text{// Forward diffuse motif}}
    \State $\dt{\breve \X_\M}{1:T} \sim q(\dt{\X_\M}{1:T} \mid \dt{\X_\M}{0})$\\
    
    \State \text{// Reverse diffuse particles }
    \State {$\forall k,\  \dt{\X_{k}}{T} \overset{i.i.d.}{\sim} \ptheta( \dt{\X}{T})$}
    \For{$t = T,\dots, 1$}
        \State{\text{// \emph{Replace} motif} }
        \State $\forall k,\  \dt{\X_{k}}{t} \leftarrow [\dt{\breve \X_{\M}}{t}, \dt{\X_{\Sc,k}}{t}]$\\
        \State{\text{// Re-weight based on }$\dt{\breve \X_\M}{t-1}$}
        \State{$\forall k,\  \dt{w_{k}}{t}\leftarrow \ptheta(\dt{\breve \X_\M}{t-1}\mid \dt{\X_k}{t})$}
        \State{$\forall k, \ \dt{\tilde{w}_{k}}{t}\leftarrow \dt{w_{k}}{t} / \sum_{k^\prime=1}^K \dt{w_{k^\prime}}{t}$}
        %\State{$\dt{\pr_{\M,K}}{t} \leftarrow \sum_{k=1}^K \dt{\tilde w_{k}}{t} \delta_{\dt{\X_{\Sc,k}}{t}} $}
        \State $\dt{\tilde \X_{1:K}}{t} \sim \texttt{Resample}(\dt{\tilde{w}_{1:K}}{t}, \dt{ \X_{1:K}}{t})$ \\
        \State {\text{// Propose next step}}
        \State {$\forall k,\  \dt{\X_k}{t-1} \overset{indep.}{\sim} \ptheta(\dt{\X}{t-1} \mid \dt{\tilde \X_{k}}{t})$}
    \EndFor 
\State \text{Return} $\dt{\X_{\Sc,1:K}}{0} $
\end{algorithmic}
\end{algorithm}
\end{minipage}
% \end{figure}
\vspace{-40pt}
\end{wrapfigure}

To overcome this intractability, we build on the work of \citet{song2020score}, who introduced a practical algorithm that generates \emph{approximate} conditional samples.
This strategy is to (1) forward diffuse the conditioning variable to obtain $\dt{\X_\M}{1:T}\sim q(\dt{\X_\M}{1:T}\mid \dt{\X_\M}{0}),$
and then (2) for each $t,$ sample $\dt{\X_\Sc}{t} \sim \ptheta(\dt{\X_\Sc}{t} \mid  \dt{\X_\M}{t+1},\dt{\X_\Sc}{t+1}).$
We call this approach the \emph{replacement method} (following \citet{ho2022video}) and make it explicit in Appendix \Cref{alg:replacement}.
However, in \Cref{prop:KL_div} we show that the replacement method introduces irreducible approximation error that cannot be eliminated by making $p_\theta$ more expressive.
Additionally, although this approximation error is not analytically tractable in general, we show in \Cref{corollary:diverging_error} the dependence of this error on the covariance of $\dt{\X_\M}{0}$ and $\dt{\X_\Sc}{0}$ in the case that $q(\dt{\X_\M}{0}, \dt{\X_\Sc}{0})$ is bivariate Gaussian.
\subsection{Conditional sampling is a sequential Monte Carlo problem}\label{sec:smc}
We next frame approximation of $q(\dt{\X_\Sc}{0} \mid \dt{\X_\M}{0})$ as a sequential Monte Carlo problem that we may solve by particle filtering.
Intuitively, particle filtering addresses a limitation of the replacement method: the failure at each time $t$ to look beyond the current step to the less-noised motif $\dt{\X_\M}{t-1}$ when sampling $\dt{\X_\Sc}{t} \sim \ptheta(\dt{\X_\Sc}{t} \mid \dt{\X}{t+1}).$
Our key insight is that because $\ptheta(\dt{\X_\M}{t-1}\mid \dt{\X}{t})$ provides a mechanism to assess the likelihood of $\dt{\X_\M}{t-1}$, we can prioritize noised scaffolds that are more consistent with the motif.
Particle filtering leverages this mechanism to provide a sequence of discrete approximations to each $\ptheta( \dt{\X_\Sc}{t}\mid \dt{\X_\M}{t-1:T})$ that look ahead by this extra step.
Finally, at $t=0$ we have an approximation to $\ptheta(\dt{\X_\Sc}{0}\mid \dt{\X_\M}{0:T}).$
Then, using \Cref{prop:asymptotic_accuracy} below,  we can obtain an approximate sample from $q(\dt{\X_\Sc}{0}\mid \dt{\X_\M}{0}).$
This framing permits the application of standard particle filtering algorithms \citep{doucet2001sequential}.
\Cref{alg:particle} summarizes an implementation of this procedure that uses residual resampling \citep{doucet2009tutorial} to mitigate the collapse of the sequential approximations into point masses.
\samplingname{} provides a tunable trade-off between computational cost and statistical accuracy through the choice of the number of particles $K$.
In our next proposition we make this trade-off explicit.

%Let $\{X_{\Sc,k}\}_{k=1}^K$ be the samples returned by Algorithm 2 and define $p_{\M,K}:=\frac{1}{K} \sum_{k=1}^K \delta_{X_{\Sc,k}}$ be an empiric measure supported on these samples, where $\delta_{\X}$ is a Dirac measure on $\X.$
%Note that $p_{\M,K}$ is a random probability measure.

\begin{prop}\label{prop:asymptotic_accuracy}
Suppose that $\ptheta$ exactly matches the forward diffusion process 
such that for every $\dt{\X}{t+1}, \, \ptheta(\dt{\X}{t} \mid \dt{\X}{t+1}) = q(\dt{\X}{t} \mid \dt{\X}{t+1})$ 
and consider any motif $\dt{\X_\M}{0}.$
Let $\X_{\Sc, K}$ be a particle chosen at random from the output of \Cref{alg:particle} with $K$ particles.
Then $\X_{\Sc, K}$ converges in distribution to $q(\dt{\X_\Sc}{0}\mid \dt{\X_\M}{0})$
as $K$ goes to infinity.
\end{prop}

The significance of \Cref{prop:asymptotic_accuracy} is that it guarantees \Cref{alg:particle} can provide arbitrarily accurate conditional samples provided \revision{an accurate diffusion model and} large enough compute budget (determined by the number of particles).
\revision{To our knowledge, \samplingname{} is the first algorithm for asymptotically exact conditionally sampling from unconditional DPMs}. % that can provide arbitrary accuracy.
% To our knowledge, \samplingname{} is the first algorithm for conditionally sampling from DPMs that can provide arbitrary accuracy.
Our proof of the proposition, which we leave to \Cref{sec:smc_appendix}, is obtained from an application of standard asymptotics for particle filtering \citep[Proposition 11.4]{chopin2020introduction}.

% \begin{wrapfigure}[11]{r}{0.5\textwidth}

\section{Experiments}\label{sec:experiments}
We empirically demonstrate the ability of our method to scaffold motifs and sample protein backbone structures.
We describe our procedure for evaluating backbone designs in \Cref{sec:sc}.
We demonstrate the promise of our method for the motif-scaffolding problem in \Cref{sec:inpainting}. And we investigate our method's strengths and weaknesses via experiments in unconditional sampling in \Cref{sec:unconditional_sampling}. 
We train a single instance of \modelname{} and use it across all of our experiments. 
For simplicity, we limited our training data to single chain proteins taken from PDB that are no longer than 128 residues.
See \Cref{sec:training} for training details.
% Code for \modelname{} and \samplingname{} is included in the supplementary material and will be made available through github.
% We intend to make a later version of \modelname{} and \samplingname{} available on github.

\textbf{Baselines.}
As mentioned in \Cref{sec:related}, \citet{wang2021deep} is the only prior machine learning work to address the motif-scaffolding problem.
We do not compare against this as a baseline because no stable implementation was available at the time of writing.
The most closely related method for unconditional sampling with available software is trDesign \citep{anishchenko2021novo}, but this method does not allow specification of a motif.
The ML method most similar to \modelname{} is the concurrently developed equivariant diffusion model (EDM) proposed by \citet{hoogeboom2022equivariant}.
Like \modelname, EDM uses a densely connected EGNN architecture but without sequence-distance edge features.
Consequently, it does not impose any sequence order, and therefore does not yield a way to relate generated coordinates to a backbone chain.
%, and we found empirically fails to fit small backbones and does not provide coherent chain structure.

\subsection{\emph{In silico} evaluation of designed backbones}
\label{sec:sc}
While experimental validation via X-ray crystallography remains the gold standard for evaluating computationally designed proteins,
recent work \citep{wang2021deep,lin2021deep} has proposed to leverage highly accurate protein structure prediction neural networks as an \emph{in silico} proxy for true structure.
More specifically, \citet{wang2021deep} jointly design protein sequence and structure, and validate by comparing the design and \texttt{AlphaFold2} (AF2) \citep{Jumper2021HighlyAP} predicted structures. 
%Protein designs ultimately need to undergo experimental validation to assess their fitness and function, but this is costly and time-consuming. 
% such as pLDDT which is believed to correlate with stable proteins \cite{tunyasuvunakool2021highly}.
% We report pLDDT of our designs in \cref{sec:additional_metrics}.
Here, our goal is to assess the quality of scaffolds generated independent of a specific sequence,
so we treat fixed backbone sequence design as a downstream step as in \citet{lin2021deep}.

Our evaluation with AF2 is as follows. For each generated scaffold we use a C-$\alpha$ only version of \protmpnn{}  \citep{dauparas2022protmpnn}
with a temperature of 0.1 to sample 8 amino acid sequences likely to fold to the same backbone structure.
We then run AF2 with the released CASP14\footnote{
Biannual protein folding competition where AF2 achieved first place. Weights available under Apache License 2.0 license.}
weights and 15 recycling iterations.
We do not include a multiple sequence alignment as an input to AF2.
\revision{Our choice of utilizing \protmpnn{} and AF2 (without MSAs) is motivated by their empirical success in various de novo protein design tasks and the ability to recapitulate native proteins \citep{dauparas2022protmpnn, bennett2022improving}}.
To assess unconditionally sampled scaffolds, we then evaluate the agreement of our backbone sample with the AF2 predicted structures using the maximum TM-score \citep{zhang2005tm} across all generated sequences which we refer to as \scTM, for \emph{self-consistency} TM-score.
To assess whether prospective scaffolds generated support a motif, we compute the root mean squared distances (RMSD) of the desired and predicted motif coordinates after alignment and refer this metric as the \emph{motif RMSD}.
Appendix \cref{alg:self_consistency} outlines the exact steps.

Because a TM-score > 0.5 indicates that two structures have the same fold \citep{zhang2005tm},
we say that a backbone is designable if \scTM{} > 0.5. The ability for AF2 to reproduce the same backbone from an independently designed sequence is evidence a sequence can be found for the starting structure. \revision{To verify this is a reasonable cutoff, we analyzed \scTM{} over our training set and found 87\% to be designable.}

% A more thorough discussion of \scTM{} is presented in \cref{sec:additional_metrics} \TODO{Do we want to promise this?}.
%In addition, many proteins contain disordered regions that would incur high RMSD between two conformations of the same sequence. Our \scRMSD calculation cannot account for this and lead to many false negatives of designable backbones. On one hand, it is desirable to design hyperstable proteins, but prohibits from designing disordered binding regions. We leave this as a direction of future work.

% \BLT{add explanation of how (in constrast to Doug & Jue's paper, FB paper, which use this post-hoc to assess if design was "on target", we use this as a systematic evaluation criterion for methods / populations of generated backbones.}

% \TODO{May move to appendix} At a high-level, \scRMSD \ calculation can be thought of as a noisy channel: ProtMPNN attempts to compress the 3D backbone into a 1D sequence which AF2 uses to recover the original 3D backbone. We argue a low \scRMSD (defined as <3 \AA) would indicate the backbone is \emph{designable}, a sequence can be found that would fold into the desired backbone structure. Our argument relies on the fact most random sequences do not fold into a structure and physically impossible backbones would not map to a sequence \cite{anishchenko2021novo}. The ability of ProtMPNN to infer a sequence for a backbone and then have it be recapitulated by AF2 is statistically significant especially considering ProtMPNN and AF2 were independently trained.

\subsection{Motif-scaffolding via conditional sampling}
We evaluated our motif-scaffolding approach (combining \samplingname{} and \modelname{}) on motifs extracted from existing proteins in the PDB and found that our approach can generate long and diverse scaffolds that support these motifs.
\begin{figure}
\includegraphics[width=1.0\textwidth]{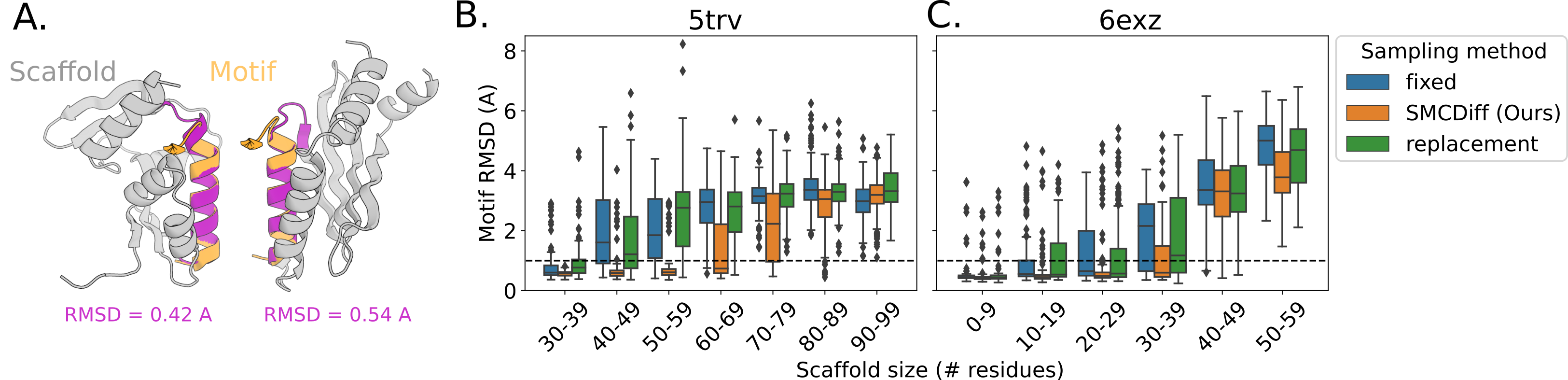}\
\caption{
Motif-scaffolding case studies.
    (A) Example of two scaffold structures generated around a segment of \texttt{5trv}.  Orange: desired input motif, Grey: AlphaFold-predicted structure of two scaffolds, with the motif highlighted (purple).  Both scaffolds were sampled using \samplingname{} with \scTM{} > 0.5. 
    (B,C) Motif RMSD for 5trv and 6exz test cases, its dependence on scaffold size, and comparison of \samplingname{} to two naive inpainting methods (\texttt{fixed}, \texttt{replacement}).
}\label{fig:inpainting}
\end{figure}
We chose to first evaluate on motifs extracted from proteins present in the training set because we knew that at least one stabilizing scaffold exists.
We considered 2 examples 
taken from the PDB with IDs 6exz and 5trv, which are 69 and 118 residues long, respectively.
\revision{We chose these examples due to their high secondary structure composition while being representative of the shortest and longest lengths seen during training.}
For each \revision{structure}, we chose a 15--25 \revision{residue helical} segment as the motif (see \Cref{sec:additional_results} for details).
The remainder of each protein is one possible supporting scaffold.
We sought to assess if we could recover this and other scaffolds with the same size and motif placement.

Based on prior work \citep{wang2021deep},
we expected that building larger scaffolds around a motif would be more challenging than building smaller scaffolds.
To assess this length dependence, we expanded the segment of used as the motif when running \samplingname{} by including additional residues on each side.
In each case, though, we compute the motif RMSD over the minimal motif.
In \Cref{fig:inpainting}B, we present motif-scaffolding performance and its dependence on scaffold size for 5trv, the longer of the two test proteins.
For the 5trv test case, the lower quartile of the motif RMSD for \samplingname{} is below 1\AA{} for scaffolds up to 80 residues.
Since 1\AA{} is atomic-level resolution, we conclude that our approach can succeed in this length range.

\Cref{fig:inpainting}A provides a visualization of our method's capacity to generate long and diverse scaffolds. 
The figure
depicts two dissimilar scaffolds of lengths 34 and 54 produced by \samplingname{} with 64 particles.
Both scaffolds are designable and agree with AF2 (\scTM{} > 0.5).
Diversity is particularly evident in the different orderings of secondary structures.

\Cref{fig:inpainting}B compares \samplingname{} to two naive inpainting methods, \texttt{fixed} and \texttt{replacement}.
In \texttt{fixed}, the motif is fixed for every timestep $t,$ and the reverse diffusion is applied only to the scaffold (as done by \citet{zhou20213d});
\texttt{replacement} is the method described in \Cref{sec:conditional_sampling}.
In contrast to \samplingname{}, these baselines fail to generate a successful scaffolds longer than 50 residues on 5trv, as determined by the location of their lower quartiles.

We next applied these three inpainting methods to harder target\revision{s in order to measure generalization to out-of-distribution and more difficult motifs comprising of dis-contiguous regions and loops}.
We consider a motif obtained from the respiratory syncytial virus (RSV) protein \revision{ and calcium binding EF-hand motif, both of which are not in the training dataset. RSV is known to be difficult due to its composition of helical, loop, and sheet segments, while EF-hand is a dis-contiguous loop motif found in a calcium binding protein. More details about both motifs can be found in \cite{wang2021deep}; there the authors report the only known successful scaffold of these motifs but they attain it with a computationally intensive \emph{hallucination} approach.}
% This motif has been successfully scaffolded by the  of \citep{wang2021deep}.
We found that our method failed to generate scaffolds predicted to recapitulate the motif (\Cref{sec:additional_results})\revision{; however, \samplingname{} provided smaller median motif RMSDs than the other two inpainting methods.}
% Our results indicate that our method does not improve uniformly upon this previous approach.

\textbf{Compute cost.} The computation of \samplingname{} with 64 particles is approximately 2 minutes per independent sample, while alternative methods \texttt{fixed} and \texttt{replacement} can produce 64 independent samples in the same time.
By contrast, the hallucination approach of \cite{wang2021deep} involves running a Markov chain for thousands of steps, and has runtime on the order of hours for a single sample \citep{anishchenko2021novo}.
\label{sec:inpainting}

\subsection{Unconditional sampling}
\begin{figure}
\centering
\includegraphics[width=1.0\textwidth]{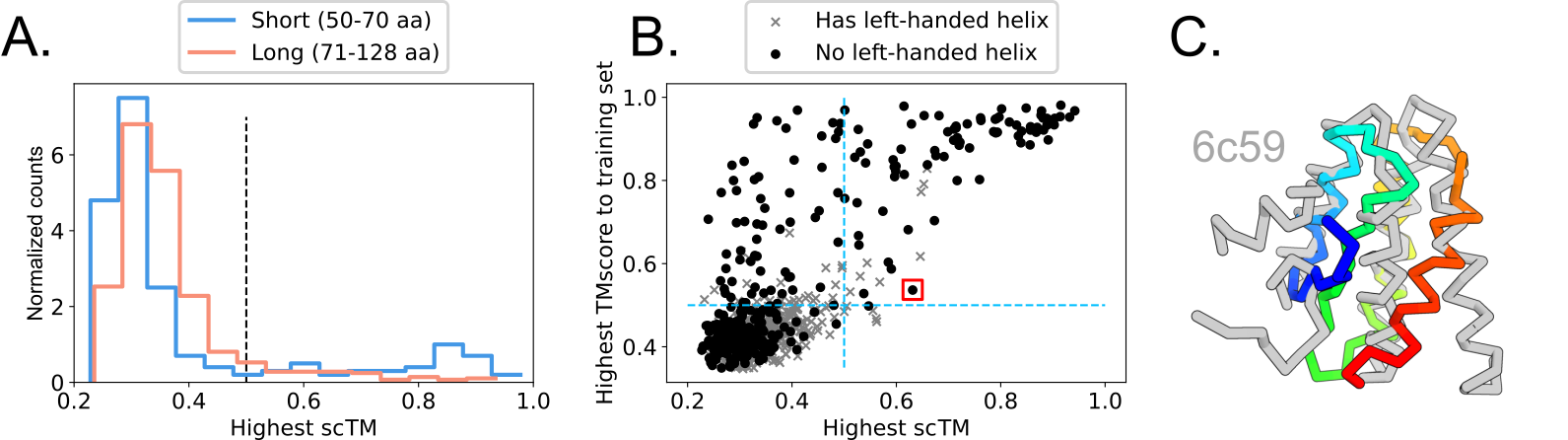}\
    \caption{
    Protein backbone samples from \modelname. (A) Density plot of \scTM{} for different length categories (50--70, 70--128). The dashed line at \scTM{} = 0.5 indicates the threshold of "designability", points to the right are considered "designable" (see text).  (B) Scatter plot of \scTM{} and the highest TM-score of each sample to all of PDB. Points represented as a grey ``$\times$'' are detected to contain an (invalid) left-handed helix. Dashed lines indicate thresholds \scTM{} = 0.5.
    (C) Example of a designable backbone sample (rainbow) with \scTM{} > 0.5 (boxed in red in panel B) to its closest PDB example (6c59, grey) with a TM-score of 0.54. 
}\label{fig:unconditional}
\end{figure}
We next investigate the origins of the diversity seen in \Cref{fig:inpainting} by analyzing the diversity and designability of \modelname{} samples without conditioning on a motif. 

We first check that \modelname{} produces designable backbones.
To do this, we generated 10 backbone samples for each length between 50 and 128 and then calculated \scTM{} for each sample.
In \cref{fig:unconditional}A, we find that 11.8\% of samples have \scTM{} > 0.5.
However, the majority of backbones do not pass this threshold.
We also observe designability has strong dependence on length since we expect that longer proteins are harder to model in 3D and design sequences for.
We separated the lengths below 128 residues into two categories and refer to them as \emph{short} (50--70) and \emph{long} (70--128).
Our results in \Cref{fig:unconditional}A indicate 17\% of designs in the short category are designable vs.\ 9\% in the long category.
In \Cref{fig:clustering}, we present a structural clustering of these designable backbones;
we find that these backbones exhibit diverse topologes. 
% In \Cref{sec:additional_results}, we find that this trend agrees with a secondary measure of performance, using RMSD instead of TMscore.

We next sought to evaluate the ability of \modelname{} to generalize beyond the training set and produce novel backbones.
In \Cref{fig:unconditional}B each point represents a backbone sample from \modelname. The horizontal coordinate of a point is the \scTM{}, and the vertical coordinate is the minimum TM-score across the training set.
%RB I think "strongly" should be "strong"
We found a strong positive correlation between \scTM{} and this minimum TM-score, indicating that many of the most designable backbones generated by \modelname{} were a result of training set memorization.
However, if the model were only memorizing the training set, we would see TM-scores consistently near 1.0; the range of scores in \Cref{fig:unconditional}B indicate this is not the case -- and the model is introducing a degree of variability.
\Cref{fig:unconditional}C gives an example of backbone with \scTM{} > 0.5 that appears to be novel. 
Its closest match in the PDB has TM-score = 0.54.

\cref{fig:unconditional}B illustrates a limitation of our method: many of our sampled backbones are not designable.
One contributing factor is that \modelname{} does not handle chirality. Hence \modelname{} generates backbones with the wrong handedness, which cannot be realized by any sequence.
\cref{fig:unconditional}B shows that 45\% of all backbone samples had at least one incorrect, left-handed helix.
Of these, most have \scTM{} < 0.5.
We describe calculating left-handed helices in 
\cref{sec:additional_metrics}.

\cref{fig:latent_interpolation} illustrates an interpolation between two samples, showing how \modelname's outputs change as a function of the noise used to generate them.
To generate these interpolations, we pick two backbone samples that result in different folds.
For independent samples generated with noise $\dt{\epsilon}{0:T}$ and $\dt{\tilde \epsilon}{0:T}$ we interpolate with noise set to $\sqrt{\alpha}\dt{\epsilon}{0:T} + \sqrt{1-\alpha}\dt{\tilde \epsilon}{0:T}$ for $\alpha$ between 0 and 1.
The depicted values of $\alpha$ are chosen to highlight transition points with full interpolations included in \Cref{sec:additional_interpolations}.
A future direction is to exploit the latent structure of \modelname{} to control backbone topology.

\begin{figure}
\centering
\includegraphics[width=1.0\textwidth]{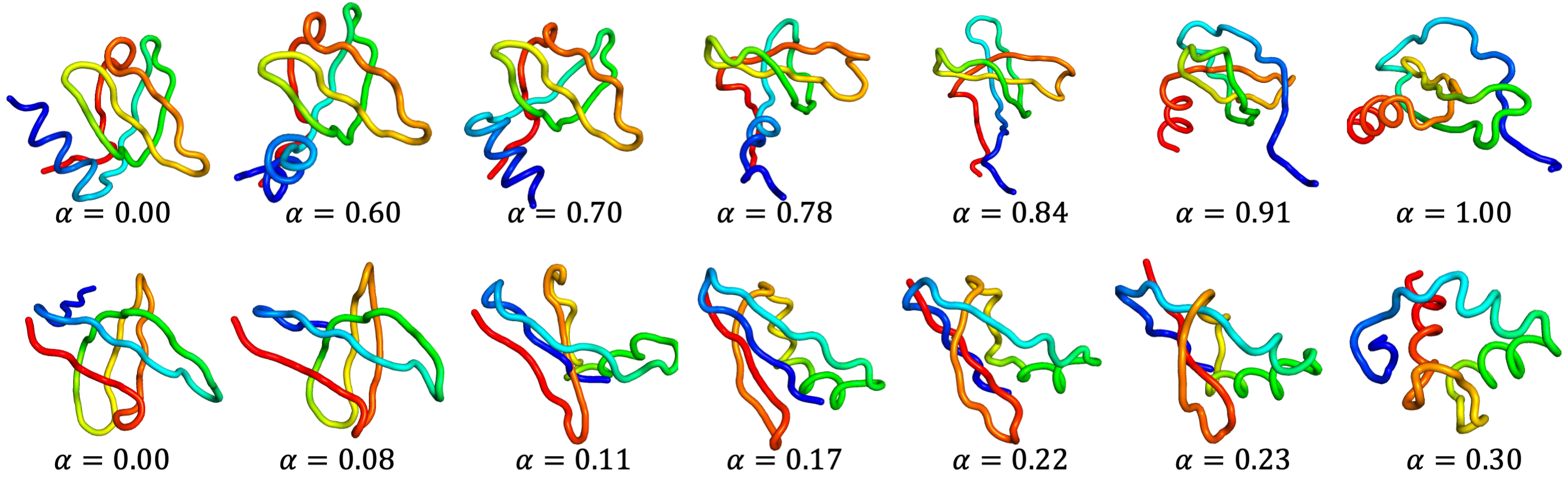}
    \caption{
    Interpolations between \modelname{} samples demonstrating the diversity of backbones captured.
    Top: 64-residue example.  Bottom: 56-residue example.
    \modelname{} samples are determined by the Gaussian noise across all steps, $\dt{\epsilon}{0:T}.$
}\label{fig:latent_interpolation}
\end{figure}

\label{sec:unconditional_sampling}

% \subsection{Dataset details}\label{sec:data}
% \input{data}

\section{Discussion}\label{sec:discussion}
The motif-scaffolding problem has applications ranging from medicine to material science \citep{king2012computational}, but
remains unsolved for many functional motifs.
We have created the first generative modeling approach to motif-scaffolding by developing \modelname{}, a diffusion probabilistic model of protein backbones,
and \samplingname{}, a procedure for generating scaffolds conditioned on a motif.
Although our experiments were limited to a small set of proteins, our results demonstrate that our procedure is the first capable of generating diverse scaffolds longer than 20 residues with computation time reliably on the order of minutes or less.
Our work demonstrates the potential of machine learning methods to be applied in realistic protein design settings.

\revision{\textbf{General conditional sampling}. \samplingname{} is applicable to generic DPMs and is not limited to only proteins and motif-scaffolding.
While we do not make claims of \samplingname{} outperforming state-of-the-art conditional diffusion models on other tasks such as image generation, we demonstrate a clear advantage of \samplingname{} over the replacement method on a toy task of inpainting MNIST images in \cref{sec:mnist_inpaint}.
Extending \samplingname{} outside of motif-scaffolding is outside the scope of the present work, but the advantages of a single model for both unconditional and conditional generation warrants additional research.}

\textbf{Modeling limitations}. Our present results do not indicate our procedure can generalize to motifs that are not present in the training set.
We believe improvements in protein modeling could provide better inductive biases for generalization.
\revision{\modelname{}, based on EGNN, is reflection equivariant since it only sees pairwise distances between 3D C-$\alpha$ coordinates.}
% For example, \modelname{} sees only 3D backbone coordinates and is blind to the the orientation of backbone residues
% As a consequence, \modelname{} is reflection symmetric and generates left-handed helices (\Cref{sec:chirality}), which do not stably occur in natural proteins.
%Designing a generative model capable of detecting chirality is a non-trivial task and is not the focus of our work.
Additionally, \modelname{} does not explicitly model primary sequence or side-chains.
\citet{hoogeboom2022equivariant} demonstrate the benefits of modeling sequence information in small molecules; joint modeling sequence and structure in a single model could improve the designability of protein scaffolds and backbones as well.

\textbf{Data limitations}. We remarked our training set is small due to filtering based on length and oligometry (using only monomeric proteins).
Scaling up to longer proteins opens up thousands more examples from the PDB, but in preliminary experiments has proven challenging.
Lastly, further development and comparison of methods for motif scaffolding will benefit from standard evaluation benchmarks.
\revision{Developing a benchmark proved to be difficult since motifs are not labeled in protein databases. It will be important to gather motifs of biological importance in order to guide ML method development towards real-world applications.}
Because no such benchmarks exist, developing them is a valuable direction for future work.

% \textbf{Ethical considerations.} Protein engineering has potential negative societal impact through advancing biotechnology.
% Our work and results address a computational problem that could be used in beneficial or harmful applications of biotechnology.

\subsubsection*{Acknowledgements}
The authors thank Octavian-Eugen Ganea, Hannes St\"{a}rk, Wenxian Shi, Felix Faltings, Jeremy Wohlwend, Nitan Shalon, Gabriele Corso, Sean Murphy, Wengong Jin, Bowen Jing, Renato Berlinghieri, John Yang, and Jue Wang for helpful discussion and feedback.
We thank Justas Dauparas for access to an early version of \protmpnn{}.
We dedicate this work in memory of Octavian-Eugen Ganea who initiated the project by connecting all the authors. 

BLT and JY were supported in part by NSF-GRFP.
BLT and TB were supported in part by NSF grant 2029016 and an ONR Early Career Grant. JY, RB, and TJ acknowledge support from NSF Expeditions grant (award 1918839: Collaborative Research: Understanding the World Through Code), Machine Learning for Pharmaceutical Discovery and Synthesis (MLPDS) consortium, the Abdul Latif Jameel
Clinic for Machine Learning in Health, the DTRA Discovery of Medical Countermeasures Against
New and Emerging (DOMANE) threats program, the DARPA Accelerated Molecular Discovery
program and the Sanofi Computational Antibody Design grant.
DT and DB were supported with funds provided by a gift from Microsoft. DB was additionally supported by the Audacious Project at the Institute for Protein Design, the Open Philanthropy Project Improving Protein Design Fund, an Alfred P. Sloan Foundation Matter-to-Life Program Grant (G-2021-16899) and the Howard Hughes Medical Institute.

\bibliography{references.bib}
\bibliographystyle{iclr2023_conference}

%\section*{Checklist}
%\input{checklist}

\newpage

\appendix
\addcontentsline{toc}{section}{Appendix} % Add the appendix text to the document TOC
\part{Appendix} % Start the appendix part
\parttoc % Insert the appendix TOC

% \section{Background}\label{sec:background}
% \input{supplement/background}

% \section{Additional Related work}\label{sec:supp_related_work}
% \input{supplement/more_related_work}
% \tableofcontents

\section{Additional related work}\label{sec:additional_related}
We next cover additional related work on generative models of proteins sequence and structure, beyond the discussion in \cref{sec:related}.
Following the success of deep language models, 
\citet{ferruz2022deep} developed protein sequence models to generate new proteins, but these models do not allow specification of structural motifs.
Another class of methods, referred to as fixed backbone sequence design \citep{fleishman2011rosettascripts,ingraham2019generative,xiong2020increasing,mcpartlon2022deep,hsu2022learning}, attempts to solve the problem of identifying a sequence that folds into any given designable backbone structure.
In the present work, we utilize a particular sequence design method, \protmpnn{}  \citep{dauparas2022protmpnn}, but in principle any other fixed-backbone sequence design method could be used in its place.
\citet{anand2018generative, lin2021deep, wu2021ebmfold} propose generative adversarial networks, variational autoencoders, and energy-based models, respectively, on distance matrices, but these approaches 
(1) do not generate backbones compatible with a specified motif
and (2) rely on an unwieldy optimization step to translate the distance matrix into backbone coordinates.
Other authors use neural net \citep{tischer2020design,anishchenko2021novo,wang2021deep,huang2022backbone, wu2021ebmfold},
but require a computationally challenging conformational landscape exploration.

\section{Problem assumptions and modeling heuristics}\label{sec:supp_caveats}
The formulation of the motif-scaffolding problem presented in \Cref{sec:motif_scaffolding_formulation} makes several simplifying assumptions,
and our modeling approach relies on several heuristics.
We describe these assumptions and heuristics in what follows,
and comment on how they might be addressed by further methodological developments.
But we first describe an illustrative example of an instance of motif scaffolding.

 \revision{
 \textbf{Protein sequence-structure relationship.}
 Generally speaking, a protein's sequence encodes an ensemble of conformations, populated to different degrees at biological temperatures. Anfinsen's hypothesis states that the ground state conformation is thermodynamically accessible \citep{anfinsen1973principles}, providing a mapping from sequence to a unique (ground state) structure. In practice, the ground state structures make up the vast majority of experimentally determined protein conformations, as over 95\% of structures in the Protein Data Bank (PDB) are collected at cryogenic temperatures \citep{fraser2011accessing}. Thus we simplify our problem by saying that a sequence uniquely maps to a static structure (i.e.\ the ground state structure).
However, violations of this assumption arise in some PDB structures as a result of
(1) of context specific determinants of structure such as post-translational modifications and environmental factors including pH, binding partners, and salts, as well as 
(2) thermodynamic inaccessibility of the ground state.
% However, it is important to note ways in which the data in the PDB may not always reflect this assumption:
% \begin{itemize}
%     \item  Post-translational modifications may alter protein structure and function, which biology often takes advantage of to regulate a process. However, if we view that modification as a change to the primary sequence, our sequence-structure relationship still holds.
%     \item The ground state conformation is also affected by the environment. There may be different structures of the same protein if they are taken with different pHs, binding partners or salts.
%     \item Not all ground state structures are thermodynamically accessible. A protein may need the assistance of cellular chaperons to adopt a functional conformation. This help is not captured by the crystal structure.
% \end{itemize}
}
 
\textbf{Motif sequence and side-chains.}
As stated in \Cref{sec:motif_scaffolding_formulation}, we assume we may represent a functional motif by the coordinates of its C-$\alpha$ atoms.
However, the biochemical functions of proteins depend not only on backbone structure, but also on side-chains.
For example, the activity of many enzymes is imparted by triplets of residues, known as \emph{catalytic triads},
whose ability to catalyze reactions depends on the spatial organization of side-chain atoms.
Our problem statement and subsequent evaluation scheme are agnostic to the amino acid identity of motif residues, let alone side-chain positioning.
A more complete representation of a motif would include the side-chain identities (i.e.\ the amino acid \emph{sequence}) and side-chain atom coordinates.

\textbf{Scaffold length and motif placement.}
We have additionally assumed that the size of scaffolds and
the indices of motif residues within the backbone chain, $\M,$ are known a priori.
However, in practice satisfactory scaffolds could have different lengths and different motif placements,
and typically it is not known a priori what lengths and placements will be best.
Previous works have addressed this challenge through brute force by sampling multiple lengths and placements,
and relied on post-hoc filtering to identify the most promising scaffolds \citep{wang2021deep, yang2021bottom}.
Subsequent work on ML methods could potentially generalize beyond this assumption to efficiently sample appropriate scaffold lengths and motif placements.

\textbf{Sequence and side-chain modeling.}
\modelname{} models only the backbone coordinates
and leaves sequence design to a subsequent stage, for which we have used \protmpnn{}. 
A more complete representation of a proteins could include both sequence and structure (where structure can be divided into the backbone and side-chain atom coordinates).
To model sequence, we rely on a separately trained neural network, \protmpnn{}, but this is not ideal.
% Using a separately trained neural network is not ideal since \modelname{} and \protmpnn{} are trained independently on different training data.
Unless \modelname{} produces perfect backbones, one would expect the backbone samples of \modelname{} to present a substantial domain shift when used as input for \protmpnn{}.
% As discussed in \cref{sec:discussion}, a limitation of \modelname{} is that it presents occasional chain breaks.
% Generalizing \modelname{} to generate sequences jointly with backbone structure would enable end-to-end training and might allow more reliable sequences by eliminating this domain shift and provide greater expressivity that could mitigate these limitations. 

\textbf{3D backbone representation.}
% The minimal coordinate choice of a protein structure is the C$\alpha$ coordinates of every residue along the backbone.
In this work, we represent a protein structure using the C-$\alpha$ coordinates of every residue along the backbone.
However, this representation is coarse-grained and ignores additional backbone atomic coordinates, namely the backbone carbon and nitrogen atoms. 
% If the amino acid type of each residue was modeled, one could also model the sidechain atomic coordinates or just the C$\beta$ coordinates (which is shared among all but one amino acid).
% Our implementation of \protmpnn{} in \cref{fig:schematic} models only C$\alpha$ coordinates, but 
\citet{dauparas2022protmpnn} observed additionally modeling the heavy atoms of the backbone nitrogen and carbon atoms along with the C-$\beta$ of every residue (to capture side-chain information) improved performance (by sequence recovery) for fixed-backbone sequence design. 
We hypothesize modeling additional coordinates of every residue would also improve designability performance of \modelname{}. Constraining \modelname{} to place the remaining atoms in the correct orientation could help enforce correct chirality and mitigate chain breaks.

\section{Additional \modelname{} details}\label{sec:supp_protdiff}

As a reminder from \cref{sec:method}, each node in the graph is indexed by $n=1,\dots, N$ and corresponds to a residue with coordinates $\X_n \in \mathbb{R}^3$ and node features $h_n \in \mathbb{R}^D$. For each pair of nodes $n, n'$ we define an edge and associate it with edge features $a_{nn'} \in \mathbb{R}^D$. Our neural network to predict $\epsilon_\theta$ is an instance of EGNN composed of multiple EGCL layers . We recount details of EGCL and then discuss construction of edge and node features, $a_{nn'}$ and $h_n$.

\textbf{Equivariant graph convolution layers (EGCL).}
Each layer $l=1, \dots, L$ defines an update as $(\X^{l}, h^{l})=\text{EGCL}[\X^{l-1}, h^{l-1}]$ where for each node $n$
\begin{gather*}
    \label{eq:egnn}
    \X_n^{l} = \X_n^{l-1} + \sum_{{n^\prime} \ne n} \vec \omega_{nn^\prime}\cdot\phi_\X(h_n^{l-1}, h_{n^\prime}^{l-1}, d_{n{n^\prime}}, a_{n{n^\prime}}) \quad \text{and}\quad
    h_n^{l} = \phi_h(h_n^{l-1}, m_n), \quad \text{for}\\
    \vec \omega_{nn^\prime} = \frac{\X_n^{l-1} - \X_{n^\prime}^{l-1}}{\sqrt{d_{nn^\prime}} + \gamma}, \quad 
    m_n = \sum_{n^\prime \neq n}\phi_e(h_n^{l-1}, h_{n^\prime}^{l-1}, d_{n{n^\prime}}, a_{nn^\prime}),\text{ and } d_{nn^\prime} = \|\X_n^{l-1} - \X_{n^\prime}^{l-1}\|_2^2.
\end{gather*}
$\phi_e, \phi_h$, and  $\phi_\X$ are fully connected neural networks,
and $\gamma$ is a small positive constant included for numerical stability.
The first EGCL layer takes in initial node embeddings, $h^0$ while edge embeddings, $a_{nn'}$, are kept fixed throughout. 

\textbf{Initial node and edge embeddings.} Each edge between two residues indexed in the sequence by $(n, n')$ is featurized with $D$ features obtained through a sinusoidal encoding of its relative offset:
\begin{align*}
    a_{nn^\prime} = \begin{bmatrix}\varphi(n-n^\prime, 1)\\
    \vdots \\
    \varphi(n-n^\prime, D)
    \end{bmatrix}, \text{ where }\
    \varphi(x, k) = \begin{cases}
    \text{sin}\left(x\cdot \pi/ N^{2\cdot k/D}\right) ,& k\mod 2 = 0 \\
    \text{cos}\left(x\cdot \pi/ N^{2\cdot (k-1)/D}\right) ,& k\mod 2 = 1.
    \end{cases}
    % \label{eq:pos_feats}
\end{align*}
For node features, we similarly use a sinusoidal encoding of sequence position as well as of the diffusion time step $t$ as
\begin{align*}
    h_n(t) = \begin{bmatrix}\varphi(n, 1)\\
    \vdots \\
    \varphi(n, D)
    \end{bmatrix} + R \begin{bmatrix}\varphi(t, 1)\\
    \vdots \\
    \varphi(t, D)
    \end{bmatrix}, \
\end{align*}
where $R$ is a $D \times D$ orthogonal matrix chosen uniformly at random.
Intuitively, applying $R$ transforms the time encoding to be orthogonal to the positional encoding.

\textbf{Coordinate scaling}
\revision{While protein structures are typically parameterized in Angstroms, we transform the input protein coordinates to be in nanometers rather by dividing by 10.
This scaling brings the backbones to a spatial scale similar to the reference distribution at which the forward noising process is stationary, a unit variance isotropic Gaussian.
Importantly, the distribution of the final step $T$ is indistinguishable from an isotropic Gaussian (Supplementary \cref{fig:x_t_distribution}.)}

\begin{figure}[H]
    \centering
    \includegraphics[width=0.9\textwidth]{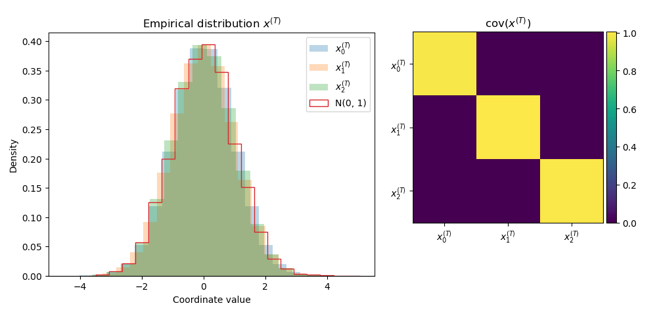}
    \caption{Distribution of $x^{(T)}$ after centering and scaling $x^{(0)}$ to nanometers.}
    \label{fig:x_t_distribution}
\end{figure}

\section{Conditional sampling: \samplingname{} details and supplementary proofs}\label{sec:smc_appendix}
We here provide additional details related to \samplingname{} and the replacement method described in \Cref{sec:conditional_sampling}.
Details of the replacement method \citep{song2020score} and our analysis of its error are in \Cref{sec:replacement_error}.
\Cref{sec:particle_details} provides details of our sampling method, \samplingname{}, including
(1) a proof of \Cref{prop:asymptotic_accuracy} 
and (2) details of the residual resampling step.
We leave technical proofs and lemmas to \Cref{sec:proofs_and_lemmas}.

%\textbf{Notation:}
%In the remainder of this section we treat $\X$ as a random variable and let $\Xval$ be the value it takes.
%We treat $q$ and $p_\theta$ as probability measures associated with the given forward and learned reverse processes, respectively.
%%$p_\theta$ is assumed to be of the form in described in \Cref{sec:preliminaries}.
%Our goal is to approximate the conditional distribution of $\X_\Sc$ given any fixed scaffold $\Xval_\M.$
%We formalize this object of interest as the (deterministic) measure $q_{\X_\Sc}(\cdot \mid \X_\M=\Xval_\M),$ 
%defined such that for all Borel measurable $A,
%q_{\X_\Sc}(A \mid \X_\M=\Xval_\M) = q(\X_\Sc\in A \mid \X_\M=\Xval_\M).$
%We analogously write $q_{\X_\Sc}(\cdot \mid \X_\M)$ and $p_{\theta, \X_\Sc}(\cdot \mid \X_\M)$ to be the random probability measures
%such that $q_{\X_\Sc}(A \mid \X_\M) = q(\X_\Sc\in A \mid \X_\M)$ and  $p_{\theta,\X_\Sc}(A \mid \X_\M)= p_\theta(\X_\Sc\in A \mid \X_\M).$
%We write other conditionals similarly.
%For notational simplicity, we overload this notation and write, for example,
%$p_{\theta, \dt{\X}{t}}(\Xval \mid \dt{\X}{t+1})$ to be the density of $\dt{\X}{t}$ evaluated at $\Xval.$
%Additionally, we introduce random variables as $Z\sim q_{\X}(\cdot \mid Y)$ if $q(Z\in A\mid Y)= q_{\X}(A\mid Y).$

\textbf{Notation.}
In the following, we require notation that is more precise than in previous sections.
For each $t=0, \dots, T,$ we let $\qall{t}{\cdot}$ and  $\pall{t}{\cdot}$ denote the density functions of $\dt{\X}{t}$ according to the forward process and to our neural network approximation of the reverse process, respectively.
We denote densities restricted to the motif and scaffold with subscripts $\M$ and $\Sc.$
For example, we here write $\pmotif{t}{\dt{\X_\M}{t}}$, whereas we wrote $\ptheta(\dt{\X_\M}{t})$ in the main text.
We write (random) conditional densities as $\qmotif{t}{\cdot \mid \dt{\X_\M}{t-1}}$ and write 
the (deterministic) conditional density for an observation $\dt{\X_\M}{t-1}=\Xval_\M$ as
$\qmotif{t}{\cdot \mid \dt{\X_\M}{t-1}=\Xval_\M}.$

An object of interest will be the Kullback-Leibler (KL) divergence.
We write $\KL{\qall{t}{\cdot}
}{\pall{t}{\cdot}
} := \int \qall{t}{\Xval} \log \frac{\qall{t}{\Xval}}{\pall{t}{\Xval}}d \Xval,$
where $\log(\cdot)$ is the natural (base $e$) logarithm.
We will also encounter the expected KL between conditional densities,  
which we will write as 
$\EKL{\qall{t}{\cdot\mid \dt{\X}{t-1} }
}{\pall{t}{\cdot\mid \dt{\X}{t-1} }
} := \int \qall{t-1}{\Xval} \KL{\qall{t}{\cdot\mid \dt{\X}{t-1}=\Xval}}{\pall{t}{\cdot \mid \dt{\X}{t-1}=\Xval}}d\Xval,$
where the outer expectation is taken with respect to the unconditional density associated with first argument of $\EKL{\cdot}{\cdot}.$

\subsection{The replacement method and its error}\label{sec:replacement_error}
The replacement method was proposed by \citet{song2020score} for the task of inpainting in the context of score-based generative models.
Work \citep{ho2022video} concurrent with the present paper applied the replacement method to DPMs.
Although \citet{song2020score} notes that this approach can be understood as \emph{approximate} conditional sampling,
they provide no discussion of approximation error.
We here show that the replacement method introduces irreducible error that is inherent to the forward process.
\Cref{alg:replacement} provides an explicit description of the replacement method.

\begin{figure}[!t]
\makebox[\textwidth][c]{%
\begin{minipage}{0.46\textwidth}
\begin{algorithm}[H]
\centering
\caption{
Replacement method for approximate conditional sampling
}
\label{alg:replacement}
\begin{algorithmic}[1]
    \State{\textbf{Input:} $\dt{\X_\M}{0}$ (motif)}
    \State{\text{// Forward diffuse motif}}
    \State $\dt{\breve \X_\M}{1:T} \sim q(\dt{\X_\M}{1:T} \mid \dt{\X_\M}{0})$\\
    
    \State \text{// Reverse diffuse scaffold} 
    \State $\dt{\X}{T} \sim \ptheta(\dt{\X}{T})$
    \For{$t = T,\dots, 1$}
        \State{\text{// \emph{Replace} with forward diffused motif}}
        \State $\dt{\X}{t} \leftarrow [\dt{\breve \X_{\M}}{t}, \dt{\X_{\Sc}}{t}]$\\
        \State {\text{// Propose next step}}
        \State {$\dt{\X}{t-1} \sim \ptheta(\dt{\X}{t-1} \mid \dt{\X}{t})$}
    \EndFor 
    \State \text{Return } $\dt{\X_\Sc}{0},\ \dt{\X}{1:T}$
\end{algorithmic}
\end{algorithm}
\end{minipage}
}
\hfill
\end{figure}

The first return of \Cref{alg:replacement}, $\dt{\X_\Sc}{0},$
is used as a putative inpainting solution or approximate conditional sample.
But \Cref{alg:replacement} additionally returns subsequent time steps, 
$\dt{\X}{1:T}.$
We denote the approximation over all steps implied by the generative procedure in \Cref{alg:replacement}
by $\preplace{1:T}{\cdot\mid \dt{\X_\M}{0}=\Xval_\M}$ and compare it to the exact conditional,
$\qall{1:T}{\cdot\mid \dt{\X_\M}{0}=\Xval_\M}$.
We here consider error in KL divergence because it permits an analytically tractable and transparent analysis.
We additionally consider the idealized scenario where $\pall{0:T}{\cdot}$ perfectly captures the reverse process.
Under this condition, the forward KL takes a surprisingly simple form.

\begin{prop}
\label{prop:KL_div}
Suppose that $\pall{0:T}{\cdot}$ exactly matches the forward diffusion process 
such that for every $\Xval, \, \pall{t}{\cdot \mid  \dt{\X}{t+1}=\Xval} = \qall{t}{\cdot \mid \dt{\X}{t+1}=\Xval}$. 
Then for any motif $\Xval_\M,$
\begin{align}
\begin{split}
\label{eqn:KL_Song}
&\KL{
\qall{1:T}{\cdot\mid \dt{\X_\M}{0}=\Xval_\M}
}{
\preplace{1:T}{\cdot\mid \dt{\X_\M}{0}=\Xval_\M}
}\\
&=
\sum_{t=1}^{T-1} \EKL{
    \qscaf{t}{\cdot\mid \dt{\X}{t+1}, \dt{\X_\M}{0}=\Xval_\M}
    }{
    \qscaf{t}{\cdot\mid \dt{\X}{t+1}}}.
\end{split}
\end{align}
\end{prop}

\Cref{prop:KL_div} reveals that the replacement method introduces approximation error that is intrinsic to the forward process and cannot be eliminated by making $\pall{0:T}{\cdot}$ more expressive.
Although the individual terms in the right hand side of \Cref{eqn:KL_Song} are not analytically tractable in general,
in the following corollary we show that this approximation error can be non-trivial by considering a special case.
For this following example, we depart from the earlier assumption that $\X$ is in 3D, and consider scalar valued $\X_M$ and $\X_\Sc.$

\begin{corollary}\label{corollary:diverging_error}
Suppose $[\dt{\X_\M}{0},\dt{\X_\Sc}{0}]$ is bivariate normal distributed with mean zero, unit variance, and covariance $\rho.$
Further suppose that 
$\qscaf{t}{\cdot  \mid \dt{\X_\Sc}{0}}=\mathcal{N}(\cdot ; \sqrt{\dt{\bar \alpha}{t}} \dt{\X_\Sc}{0}, 1-\dt{\bar \alpha}{t})$
and 
$\qscaf{t+1}{\cdot \mid \dt{\X_\Sc}{t}} = \mathcal{N}(\cdot ; \sqrt{1-\dt{\beta}{t+1}} \dt{\X_\Sc}{t}, \dt{\beta}{t+1})$
as in \Cref{sec:preliminaries}, where $\dt{\beta}{t+1}$ and $\dt{\bar \alpha}{t}$ are between 0 and 1.
Then 
$$
\EKL{
\qscaf{t}{\cdot \mid \dt{\X_\Sc}{t+1}, \dt{\X_\M}{0}}
}{
\qscaf{t}{\cdot \mid \dt{\X_\Sc}{t+1}}
}
\ge -\frac{1}{2}\left(
\log(1-\dt{\beta}{t+1} \dt{\bar \alpha}{t}\rho^2 ) + \dt{\beta}{t+1} \dt{\bar \alpha}{t} \rho^2 
\right).
$$
\end{corollary}

We note two takeaways of \Cref{corollary:diverging_error}.
First, as we might intuitively expect, this error can be large when significant correlation in the target distribution is present.
Second, we see that the approximation error can be larger at earlier time steps, when $\dt{\bar \alpha}{t}$ is closer to 1.

\subsection{\samplingname{} details and verification proof of \Cref{prop:asymptotic_accuracy}}\label{sec:particle_details}
The idea behind the \samplingname{} procedure in \Cref{alg:particle} is to break sampling of
$\dt{\X_\Sc}{0}\sim \qscaf{0}{\cdot\mid \dt{\X_\M}{0}}$
into three stages:
\begin{enumerate}
\item {Draw
$\dt{\X_\M}{1:T} \sim \qmotif{1:T}{\cdot \mid \dt{\X_\M}{0}}.$}
\item {Draw $\dt{\X_\Sc}{1:T} \sim \qscaf{1:T}{\cdot \mid \dt{\X_\M}{0:T}}.$}
\item{Draw
$\dt{\X_\Sc}{0} \sim \qscaf{0}{\cdot \mid \dt{\X_\M}{0:T}, \dt{\X_\Sc}{1:T}}$}
%$\dt{\X_\Sc}{0} \sim \qscaf{0}{\cdot \mid \dt{\X_\M}{0:T}}=\int \qscaf{0}{\cdot \mid \dt{\X_\M}{1}, \dt{\X_\Sc}{1}=\Xval} Q_1(\Xval)d\Xval.$}
\end{enumerate}
If all three steps were performed exactly, by the law of total probability $\dt{\X_\Sc}{0}$ in step (3) would (marginally) be an exact sample from $\qscaf{0}{\cdot \mid \dt{\X_\M}{0}}.$
As such, \samplingname{} aims to perform step (1) and approximate steps (2) and (3).
Step (1) corresponds to forward diffusing the motif in lines 2--3 and is exact because we diffuse according to $q$.

Step (3) corresponds to line 17 in the last iteration (when $t=1$).
Specifically, to sample from $\qscaf{0}{\cdot \mid \dt{\X_\M}{0:T}, \dt{\X_\Sc}{1:T}}$ we make three observations. 
(i) The Markov structure of the forward process implies that $\qscaf{0}{\cdot \mid \dt{\X_\M}{0:T}, \dt{\X_\Sc}{1:T}}=\qscaf{0}{\cdot \mid \dt{\X_\M}{0:1}, \dt{\X_\Sc}{1}}$.
(ii) By the assumption that the forward and approximated reverse process agree, we have 
$\qscaf{0}{\cdot \mid \dt{\X_\M}{0:1}, \dt{\X_\Sc}{1}}=\pscaf{0}{\cdot \mid \dt{\X_\M}{0:1}, \dt{\X_\Sc}{1}}$.
(iii) Finally, because $\pall{t}{\cdot \mid \dt{\X}{t+1}}$ factorizes across $\M$ and $\Sc$ for each $t$, $\pscaf{0}{\cdot \mid \dt{\X_\M}{0:1}, \dt{\X_\Sc}{1}}=\pscaf{0}{\cdot \mid \dt{\X_\M}{1}, \dt{\X_\Sc}{1}}$.
As a result, under the assumptions of the proposition, we may sample from $\qscaf{0}{\cdot \mid \dt{\X_\M}{0:T}, \dt{\X_\Sc}{1:T}},$ and perform step (3) exactly as well.

Step (2) is the only non-trivial step, and cannot be performed exactly.
The challenge is that although the reverse process approximation, $\pscaf{1:T}{\cdot \mid \dt{\X_\M}{0:T}},$ is well-defined,
computing it explicitly involves an intractable, high-dimensional integral.

%In particular, for any $t$ and $\Xval_\Sc$ we may write 
%\(
%\qscaf{1}{\Xval_\Sc\mid \dt{\X_\M}{0:T}} &= 
%\frac{
%\qscaf{t}{\Xval_\Sc\mid \dt{\X}{t+1:T}} 
%\qmotif{t-1:t}{\dt{\X_\M}{t-1:t}\mid \dt{\X}{t+1:T}, \dt{\X_\Sc}{t}=\Xval_\Sc}
%}{
%\qmotif{t-1:t}{\dt{\X_\M}{t-1:t}\mid \dt{\X}{t+1:T}}   
%}\\
%&\propto \qscaf{t}{\Xval_\Sc\mid \dt{\X}{t+1}} 
%\qmotif{t-1}{\dt{\X_\M}{t-1}\mid \dt{\X_\M}{t}, \dt{\X_\Sc}{t}=\Xval_\Sc} \\
%&= \pscaf{t}{\Xval_\Sc\mid \dt{\X}{t+1}} 
%\pmotif{t-1}{\dt{\X_\M}{t-1}\mid \dt{\X_\M}{t}, \dt{\X_\Sc}{t}=\Xval_\Sc},
%\)
%where in the first line we apply Bayes' rule,
%in the second line we simplify using assumed conditional independences and drop multiplicative constants that do not depend on $\Xval_\Sc$,
%and in the third line we use the assumption that $\pall{t}{\cdot\mid \dt{\X}{t+1}}=\qall{t}{\cdot \mid \dt{\X}{t+1}}.$
%Therefore we can see that 
%$$\qscaf{t}{\Xval_\Sc\mid \dt{\X_\M}{t-1:T}, \dt{\X_\Sc}{t+1:T}} = 
%\pscaf{t}{\Xval_\Sc\mid \dt{\X}{t+1}} 
%\pmotif{t-1}{\dt{\X_\M}{t-1}\mid \dt{\X_\M}{t}, \dt{\X_\Sc}{t}=\Xval_\Sc}
%/Z,$$
%where $Z= \int 
%\pscaf{t}{\Xval_\Sc\mid \dt{\X}{t+1}} 
%\pmotif{t-1}{\dt{\X_\M}{t-1}\mid \dt{\X_\M}{t}, \dt{\X_\Sc}{t}=\Xval_\Sc} 
%d \Xval_\Sc,$
%is an intractable normalizing constant, and observe that exact sampling is not possible.
The sequential Monte Carlo approach of \samplingname{}, then, is to circumvent this intractability by constructing a sequence of approximations.
For each $t=T, T-1, \dots, 1,$ we approximate 
$\pscaf{t}{\cdot\mid \dt{\X_\M}{t-1:T}}$
(and thereby $\qscaf{t}{\cdot\mid \dt{\X}{t-1:T}})$
with $K$ weighted atoms (the \emph{particles}).
We denote these approximations (which are implicit in \Cref{alg:particle}) by
$\dt{\pr_K}{t}(\cdot):=\sum_{k=1}^K \dt{\tilde w_k}{t} \delta(\cdot ; \dt{\X_{\Sc,k}}{t})$,
where each $\dt{\tilde w_k}{t}$ and $\dt{\X_{\Sc,k}}{t}$ are as in \Cref{alg:particle},
and $\delta(\cdot; \X)$ denotes a Dirac mass at $\X.$
%We interpret each $\dt{\pr_K}{t}(\cdot)$ as an approximation to $\pscaf{t}{\cdot\mid \dt{\X_\M}{t-1:T}},$
%which therefore provides an approximation to each $\pscaf{t}{\cdot\mid \dt{\X_\}{t-1:T}}.$
In particular, $\dt{\pr_K}{1}(\cdot)$ is an approximation to $\pscaf{1}{\cdot\mid \dt{\X_\M}{0:T}}.$
Proving the proposition amounts to showing that in the limit as $K$ goes to infinity, each $\dt{\pr_K}{1}(\cdot)$ converges weakly to $\pscaf{1}{\cdot \mid \dt{\X_\M}{0:T}},$
which by assumption is equal to $\qscaf{1}{\cdot \mid \dt{\X_\M}{0:T}}.$
This weak convergence follows from standard asymptotics for particle filters \citep[Proposition 11.4]{chopin2020introduction},
which we make explicit in \Cref{lemma:doucet_proposition}.
As a result, if we perform step (3) with $\dt{\X_\Sc}{1}\sim \dt{\pr_K}{1}(\cdot),$
then this lemma implies that $\dt{\X_\Sc}{0}$ converges in distribution to $\qscaf{0}{\dt{\X_\Sc}{0}\mid \dt{\X_\M}{0}},$
since (i) $\qscaf{0}{\dt{\X_\Sc}{0}\mid \dt{\X_\M}{1},\dt{\X_\Sc}{1}}$
is continuous in $\dt{\X_\Sc}{1}$
and (ii) $\dt{\X_\Sc}{0}$ is independent of $\dt{\X_\M}{0}$ conditional on $\dt{\X}{1}.$

Recall that to show the proposition, it was to sufficient to show that $\dt{\pr_K}{1}$ converged weakly to 
$\qscaf{1}{\cdot \mid \dt{\X_\M}{0:T}};$
this implied that the $K$ particle returned by \Cref{alg:particle} would then converge in distribution to 
$\qscaf{0}{\cdot\mid \dt{\X_\M}{0:T}}$ which, by the law of total probability, implied that they marginally converge to $\qscaf{0}{\cdot\mid \dt{\X_\M}{0}}.$
However, while the particles return by \Cref{alg:particle} may be treated as exchangeable, they are not independent, because they depend on shared randomness in $\dt{\X_\M}{1:T}.$
To obtain approximate samples that are independent,
it is necessary to run \Cref{alg:particle} multiple times.

\textbf{Residual resampling.}
Line 14 of \Cref{alg:particle} indicates a \texttt{Resample} step.
In particle filtering, resampling steps (or \emph{branching mechanisms} \citep[Chapter 2]{doucet2001sequential}) filter out particles with very small weights,
and replace them with additional copies of particles with large weights.
Notably, the resampling step is the only point of departure of 
\Cref{alg:particle} from the replacement method;
without resampling, the algorithms behave identically.
While a variety of possible branching mechanisms exist, we use \emph{residual resampling} (\Cref{alg:residual_resample}) in our implementation for its simplicity.

\begin{figure}[!t]
\makebox[\textwidth][c]{%
\begin{minipage}{0.60\textwidth}
\begin{algorithm}[H]
%\centering
\caption{
Residual Resample
}\label{alg:residual_resample}
\begin{algorithmic}[1]
\State{\textbf{Input:} $w_{1:K}$ (weights), $\X_{1:K}$ (particles)}
 \State $\forall k, \ (c_k, r_k)\leftarrow (\lfloor K w_k \rfloor, K w_k- \lfloor K w_k \rfloor  )$
 \State $\tilde \X_C = [\underbrace{\X_1,\dots, \X_1}_{c_1}, \dots, \underbrace{\X_K,\dots, \X_K}_{c_K}]$
  \State $R\leftarrow K- \sum_{k=1}^K c_k$
 \State $[i_1, \dots, i_R] \sim \textrm{Multinomial}(r_{1:K},R)$
 \State $\tilde \X_R \leftarrow [\X_{i_1}, \dots, \X_{i_R}]$
 \State $\tilde \X = \textrm{concat}(\X_R, \X_C)$
\State \text{Return } $\tilde \X$
\end{algorithmic}
\end{algorithm}

\end{minipage}
}
\hfill
\end{figure}

\subsection{Proofs and lemmas}\label{sec:proofs_and_lemmas}
\textbf{Particle filtering lemma with technical conditions}

\begin{lemma}\label{lemma:doucet_proposition}
Consider $\dt{\pr_K}{1} := \sum_{k=1}^K \tilde w_k \delta(\cdot ; \dt{\X_{\Sc,k}}{1}),$
where $\tilde w_k$ and $\dt{\X_{\Sc,1:K}}{1}$ are as constructed in \Cref{alg:particle}.
Assume the conditions of \Cref{prop:asymptotic_accuracy}.
Then $\dt{\pr_K}{1}$ converges weakly to $\pscaf{1}{\cdot \mid \dt{\X_\M}{0:T}}$ as $K$ goes to infinity.
That is, for any Borel measurable $A,$
$\lim_{K\rightarrow \infty} \dt{\pr_K}{1}(A)= \int_A \pscaf{1}{x \mid \dt{\X_\M}{0:T}}dx$.
\end{lemma}
\begin{proof}
The proof of the lemma follows from an application of standard asymptotics for particle filtering \citep[Proposition 11.4]{chopin2020introduction}.
In particular, to apply Proposition 11.4 we use the formalism of Feynman--Kac (FK) models, following the notation of \citep[Chapter 5]{chopin2020introduction}.
Though typically (and in \citep{chopin2020introduction}) FK models are defined via a sequence of approximations at increasing time steps,
we consider decreasing time steps because we are approximating the reverse time process.
We take the initial distribution as 
$\mathbb{M}_T(\dt{\X_\Sc}{T})=\pscaf{T}{
\dt{\X_\Sc}{T}},$
the transition kernel as
$M_t(\dt{\X_\Sc}{t+1},\dt{\X_\Sc}{t})=
\pscaf{t}{\dt{\X_\Sc}{t}\mid \dt{\X}{t+1}},$
and the potential functions as 
$G_t(\dt{\X_\Sc}{t})=\pmotif{t-1}{\dt{\X_\M}{t-1}\mid \dt{\X}{t}}.$
The sequence of FK models, $\mathbb{Q}_t,$ then correspond to
$$
\mathbb{Q}_t(\dt{\X_\Sc}{t:T}) = L_t\inv\mathbb{M}_T(\dt{\X_\Sc}{T}) G_T(\dt{\X_\Sc}{T})\prod_{i=T-1}^t M_i(\dt{\X_\Sc}{i+1},\dt{\X_\Sc}{i}) G_i(\dt{\X_\Sc}{i})
$$
for each $t,$ where $L_t$ is a normalizing constant.

By substituting in our choices of $M_t$ and $G_t,$
we can rewrite and simplify $\mathbb{Q}_t$ as 
\(
\mathbb{Q}_t (\dt{\X_\Sc}{t:T})
&=  L_t\inv \pscaf{T}{\dt{\X_\Sc}{T}}\pmotif{T-1}{\dt{\X_\M}{T-1}\mid \dt{\X}{T}} 
\prod_{i=T-1}^t\pscaf{i}{\dt{\X_\Sc}{i}\mid \dt{\X}{i+1}}\pmotif{i-1}{\dt{\X_\M}{i-1}\mid \dt{\X}{i}}\\
&=  L_t \inv \pscaf{T}{\dt{\X_\Sc}{T}}\pall{t:T-1}{\dt{\X}{t:T-1}\mid \dt{\X}{T}}\pmotif{t-1}{\dt{\X_\M}{t-1}\mid \dt{\X}{t}}\\
&\propto \pall{t:T}{\dt{\X}{t:T}\mid \dt{\X_\M}{t-1}}\\
&\propto \pscaf{t:T}{\dt{\X_\Sc}{t:T}\mid \dt{\X_\M}{t-1:T}},
\)
where lines 3 and 4 drop multiplicative constants that do not depend on $\dt{\X_\Sc}{t:T}.$
From the above derivation, we see that each
$\mathbb{Q}_t (\dt{\X_\Sc}{t}) = \pscaf{t}{\dt{\X_\Sc}{t} \mid \dt{\X_\M}{t-1:T}},$
and in particular that 
$\mathbb{Q}_1 (\dt{\X_\Sc}{1}) = \pscaf{1}{\dt{\X_\Sc}{1} \mid \dt{\X_\M}{0:T}}.$
As such, the desired convergence in the statement of the lemma is equivalent to that $\dt{\pr_K}{1}$ converges to $\mathbb{Q}_1.$

\citet[Proposition 11.4]{chopin2020introduction} provide this result for the generic particle filtering algorithm (see \citet[Algorithm 10.1]{chopin2020introduction}, which is written in the FK model form described above).
More specifically, Proposition 11.4 proves almost sure convergence of all Borel measurable functions of $\dt{\pr_K}{t},$ which implies the desired weak convergence.

Although the proof provided in \citet{chopin2020introduction} is restricted to the simpler, but higher variance, case where the resampling step uses multinomial resampling, the authors note that \citet{chopin2004central} proves it holds in the case of residual resampling (which we use in our experiments) as well.
\end{proof}

\textbf{Replacement method error --- lemmas and proofs}

We here provide proofs of \Cref{prop:KL_div} and \Cref{corollary:diverging_error}.

\textbf{Proof of \Cref{prop:KL_div}:}
\begin{proof}
The result obtains from 
recognizing where the replacement method approximation agrees with the forward process,
using conditional independences in both processes,
and applying the chain rule for KL divergences.
We make this explicit in the derivation below,
with comments explaining the transition to the following line.
\(
&\KL{
\qall{1:T}{\cdot\mid \dt{\X_\M}{0}=\Xval_\M}
}{
\preplace{1:T}{\cdot\mid \dt{\X_\M}{0}=\Xval_\M}
}\\
&=\int \qall{1:T}{\dt{\Xval}{1:T}\mid \dt{\X_\M}{0}=\Xval_\M}  
\log \frac{
    \qall{1:T}{\dt{\Xval}{1:T}\mid \dt{\X_\M}{0}=\Xval_\M}
    }{
    \preplace{1:T}{\dt{\Xval}{1:T}\mid \dt{\X_\M}{0}=\Xval_\M}
    }
d \dt{\Xval}{1:T} \\
&\textrm{//\ By the chain rule of probability.} \\
&=\int \qall{1:T}{\dt{\Xval}{1:T}\mid \dt{\X_\M}{0}=\Xval_\M}  
\Big[
\log \frac{
    \qmotif{1:T}{\dt{\Xval_\M}{1:T}\mid \dt{\X_\M}{0}=\Xval_\M}
    }{
    \preplace{1:T}{\dt{\Xval_\M}{1:T}\mid \dt{\X_\M}{0}=\Xval_\M}
    } +\\
    &\ \ \ \ \ \  \ \ \ \ \ \ \ \ 
    \log \frac{
    \qscaf{1:T}{\dt{\Xval_\Sc}{1:T}\mid \dt{\X_\M}{0:T}=\dt{\Xval_\M}{0:T}}
    }{
    \preplacescaf{1:T}{\dt{\Xval_\Sc}{1:T}\mid \dt{\X_\M}{0:T}=\dt{\Xval_\M}{0:T}}
    }
\Big]\\
&\textrm{//\ By the agreement of }q \textrm{ and } p_{\replace} \textrm{ on the motif, and the chain rule of probability.} \\
&=\int \qall{1:T}{\dt{\Xval}{1:T}\mid \dt{\X_\M}{0}=\Xval_\M}  
\Big[
\log \frac{
    \qscaf{T}{\dt{\Xval_\Sc}{T}\mid \dt{\X_\M}{0:T}=\dt{\Xval_\M}{0:T}}
    }{
    \preplace{T}{\dt{\Xval_\Sc}{T}\mid \dt{\X_\M}{0:T}=\dt{\Xval_\M}{0:T}}
    } +\\
    &\ \ \ \ \ \ 
    \sum_{t=1}^{T-1}\log \frac{
    \qscaf{t}{\dt{\Xval_\Sc}{t}\mid \dt{\X_\Sc}{t+1}=\dt{\Xval_\Sc}{t+1}, \dt{\X_\M}{0:T}=\dt{\Xval_\M}{0:T}}
    }{
    \preplacescaf{t}{\dt{\Xval_\Sc}{t}\mid \dt{\X_\Sc}{t+1}=\dt{\Xval_\Sc}{t+1}, \dt{\X_\M}{0:T}=\dt{\Xval_\M}{0:T}}
    }
\Big]
d \dt{\Xval}{1:T} \\
&\textrm{//\ Because }\qscaf{T}{\cdot}=\preplacescaf{T}{\cdot}=\mathcal{N}(\cdot;0,I)
\textrm{ and the assumption that }\ptheta \textrm{ matches }q. \\
&=\int \qall{1:T}{\dt{\Xval}{1:T}\mid \dt{\X_\M}{0}=\Xval_\M}  
\Big[
    \sum_{t=1}^{T-1}\log \frac{
    \qscaf{t}{\dt{\Xval_\Sc}{t}\mid \dt{\X}{t+1}=\dt{\Xval}{t+1}, \dt{\X_\M}{0}=\dt{\Xval_\M}{0}}
    }{
    \qscaf{t}{\dt{\Xval_\Sc}{t}\mid \dt{\X}{t+1}=\dt{\Xval}{t+1}} 
    }
\Big]
d \dt{\Xval}{1:T} \\
&=\sum_{t=1}^{T-1}
\EKL{
    \qscaf{t}{\cdot \mid \dt{\X}{t+1}, \dt{\X_\M}{0}=\dt{\Xval_\M}{0}}
    }{
    \qscaf{t}{\cdot \mid \dt{\X}{t+1}}
}.
\)
\end{proof}
\textbf{Proof of \Cref{corollary:diverging_error}:}

The proof of the corollary relies of on a lemma on the variances of the two relevant conditional distributions.
We state this lemma, whose proof is at the end of the section, before continuing.
For notational simplicity, we drop the scripts and annotations on $\dt{\bar \alpha}{t}$ and $\dt{\beta}{t+1},$
and instead write $\alpha$ and $\beta,$ respectively.
\begin{lemma}\label{lemma:variances_and_bound}
Suppose $\dt{\X_\M}{0},\dt{\X_\Sc}{t},$ and $\dt{\X_\Sc}{t+1}$ are distributed as in \Cref{corollary:diverging_error}.
Then $\Var[\dt{\X_\Sc}{t}\mid \dt{\X_\Sc}{t+1}]=\beta$
and 
$\Var[\dt{\X_\Sc}{t}\mid \dt{\X_\Sc}{t+1}, \dt{\X_\M}{0}]\le \beta(1-\beta\rho^2\alpha).$
\end{lemma}

Now we provide a proof of \Cref{corollary:diverging_error}.
\begin{proof}
First recall that 
\(
\KL{
\mathcal{N}(\mu_1, \sigma_1^2)
}{
\mathcal{N}(\mu_2, \sigma_2^2)
}
&=\frac{1}{2}\left(\log\frac{\sigma_2^2}{\sigma_1^2}
+\frac{\sigma_1^2 + (\mu_1-\mu_2)^2}{\sigma_2^2} -1\right)\\
&\ge\frac{1}{2}\left(\log\frac{\sigma_2^2}{\sigma_1^2}
+\frac{\sigma_1^2}{\sigma_2^2} -1\right)
\)
and observe that this lower bound is monotonically decreasing in $\sigma_1^2$ for $\sigma_1^2 \le \sigma_2^2.$
Therefore
\(
&\EKL{
\qscaf{t}{\cdot  \mid \dt{\X_\Sc}{t+1}, \dt{\X_\M}{0}}
}{
\qscaf{t}{\cdot \mid \dt{\X_\Sc}{t+1}}
}\\
&= \int \qmotif{0}{\dt{\Xval_\M}{0}}\qscaf{t+1}{\dt{\Xval_\Sc}{t+1}\mid \dt{\Xval_\M}{0}}\Big[\\
&\ \ \ \KL{
\qscaf{t}{\cdot \mid \dt{\X_\Sc}{t+1}=\dt{\Xval_\Sc}{t+1}, \dt{\X_\M}{0}=\dt{\Xval_\M}{0}]}
}{
\qscaf{t}{\cdot \mid \dt{\X_\Sc}{t+1}=\dt{\Xval_\Sc}{t+1}]}
}\\
&\Big]d \dt{\Xval_\M}{0}\dt{\Xval_\Sc}{t+1}\\
&\ge \int \qmotif{0}{\dt{\Xval_\M}{0}}\qscaf{t+1}{\dt{\Xval_\Sc}{t+1}\mid \dt{\Xval_\M}{0}}\Big[\\
&\ \ \ \KL{
\mathcal{N}(0,\Var[\dt{\X_\Sc}{t} \mid \dt{\X_\Sc}{t+1}=\dt{\Xval_\Sc}{t+1}, \dt{\X_\M}{0}=\dt{\Xval_\M}{0}])
}{
\mathcal{N}(0, \Var[\dt{\X_\Sc}{t} \mid \dt{\X_\Sc}{t+1}=\dt{\Xval_\Sc}{t+1}])
}\\
&\Big]d \dt{\Xval_\M}{0}\dt{\Xval_\Sc}{t+1}\\
&\ge \KL{
\mathcal{N}(0,\beta(1-\beta\rho^2\alpha))
}{
\mathcal{N}(0,\beta)
}\\
&\ge \frac{1}{2}\left(
    \log \frac{\beta}{\beta(1-\beta\rho^2\alpha)}  + 
    \frac{\beta(1-\beta\rho^2\alpha)}{\beta} - 
    1
\right) \\
&= -\frac{1}{2}\left(
    \log (1-\beta\rho^2\alpha)  + 
    \beta\rho^2\alpha
\right) 
\)
where the second inequality follows from \Cref{lemma:variances_and_bound},
and the monotonicity of the KL in $\sigma_1^2.$
\end{proof}

\textbf{Proof of \Cref{lemma:variances_and_bound}:}
\begin{proof}
That $\Var[\dt{\X_\Sc}{t}\mid \dt{\X_\Sc}{t+1}]=\beta$ follows immediately from that 
$[\dt{\X_\Sc}{t},\dt{\X_\Sc}{t+1}]$ is marginally bivariate normal distributed with covariance $\sqrt{1-\beta}.$

The upper bound on $\Var[\dt{\X_\Sc}{t}\mid \dt{\X_\Sc}{t+1}, \dt{\X_\M}{0}]$ is trickier.
Observer that $[\dt{\X_\Sc}{t},\dt{\X_\Sc}{t+1}]\mid \dt{\X_\M}{0}$ is bivariate Gaussian and that 
$$
\Var\left[ \begin{bmatrix}
\dt{\X_\Sc}{t}\\
\dt{\X_\Sc}{t+1}
\end{bmatrix} \mid \dt{\X_\M}{0}
\right] = 
\begin{bmatrix}
1-\rho^2 \alpha  & \sqrt{1-\beta}(1-\rho^2\alpha) \\
\sqrt{1-\beta}(1-\rho^2\alpha) & 1 + \beta \rho^2 \alpha - \rho^2 \alpha 
\end{bmatrix}.
$$
As such, the conditional variance may be computed in closed form as 
$\Var[\dt{\X_\Sc}{t}\mid \dt{\X_\Sc}{t+1}, \dt{\X_\M}{0}]
=\beta(1-\rho^2\alpha) + (1-\beta)(1-\rho^2\alpha)\left(1- (1-\rho^2\alpha)/(1-\rho^2\alpha + \beta \rho^2\alpha) \right).$
But since $(1-\rho^2\alpha)/(1-\rho^2\alpha+\beta\rho^2\alpha)
\ge 1- (\beta\rho^2\alpha)/(1- \rho^2\alpha)$
and therefore
$1-(1-\rho^2\alpha)/(1-\rho^2\alpha+\beta\rho^2\alpha)
\le (\beta\rho^2\alpha)/(1- \rho^2\alpha)$
we can write
\(
\Var[\dt{\X_\Sc}{t}\mid \dt{\X_\Sc}{t+1}, \dt{\X_\M}{0}]
&= \beta(1-\rho^2\alpha) + (1-\beta)(1-\rho^2\alpha)\left(1- \frac{1-\rho^2\alpha}{1-\rho^2\alpha + \beta \rho^2\alpha} \right)\\
&\le \beta(1-\rho^2\alpha) + (1-\beta)(1-\rho^2\alpha)\frac{\beta\rho^2\alpha}{1- \rho^2\alpha})\\
&= \beta(1-\rho^2\alpha) + (1-\beta)\beta\rho^2\alpha\\
&= \beta(1-\beta \rho^2\alpha).
\)
\end{proof}

\section{Detecting chirality}\label{sec:chirality}
\Cref{sec:discussion} noted the limitation of \modelname{} that it can generate left-handed helices (which do not stably occur in natural proteins).
\Cref{fig:left_handed_examples} presents two such examples.
We additionally note that, as in \Cref{fig:left_handed_examples} Left, model samples can include multiple helices with differing chirality.

\begin{figure}[H]
    \centering
    \includegraphics[width=0.8\textwidth]{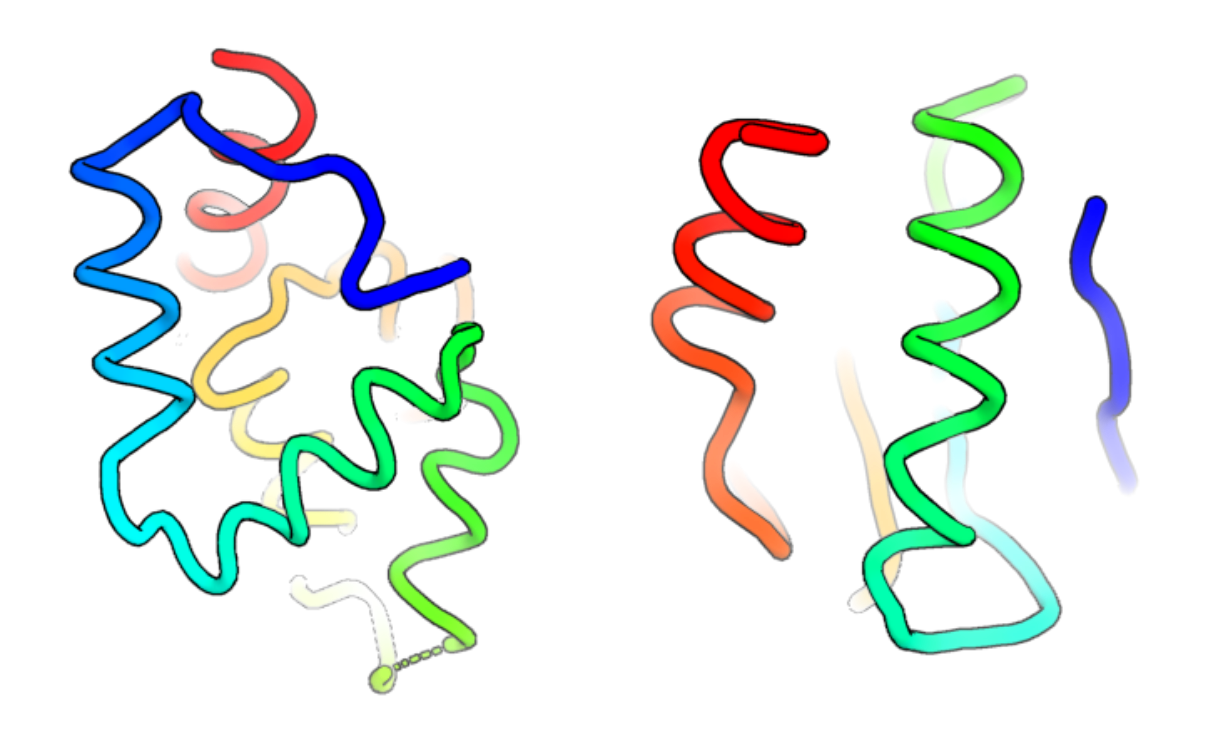}
    \caption{Two examples of protein backbone samples with incorrect left handed helices.}
    \label{fig:left_handed_examples}
\end{figure}

\section{Training details}\label{sec:training}

\modelname{} uses 4 equivariant graph convolutional layers (EGCL) with 256 dimensions for node and edge embeddings. 
The training data was restricted to single chain proteins (monomers) found in PDB and lengths in the range [40, 128].
We additionally filtered out PDB with >5\AA{} atomic resolution.
This amounted to 4269 training examples.
Training was performed using the Adam optimizer with hyperparameters \texttt{learning\_rate}=1e-4, $\beta_1=0.9$, and $\beta_2=0.999$.
We trained for 1,000,000 steps using batch size 16.
We used a single Nvidia A100 GPU for approximately 24 hours.
We implemented all models in PyTorch.
We used the same linear noise schedule as \citet{ho2020denoising} where $\beta_0 = 0.0001$, $\beta_T = 0.02$, and $T=1024$.
We did not perform hyperparameter tuning.

% For training, we restrict the data to single chain proteins (monomers) found in PDB and lengths below 128 residues. All structures with >5\AA \ resolution were removed. This amounted to 592 proteins. Next we performed self-consistency calculations over this set to understand \scRMSD for naturally occurring proteins in PDB. We then compared this distribution to self-consistency over a set of de-novo designed proteins.
% Supplementary figure XXX shows the distribution. Unsurprisingly, de-novo designed proteins exhibited much lower \scRMSD{} because they were computationally designed with similar methodologies. Natural proteins exhibit a wider variance of designability presumably due to the "marginally stable" hypothesis \cite{taverna2002proteins}. Training on non-designable proteins would give the model the incorrect inductive bias. 
% Therefore, we filtered out all proteins with >2\AA{} \scRMSD{} using TM-align which amounted to 433 proteins. A limitation of our set-up is the relatively small amount of training data. Many monomers in PDB are larger than 128 residues but scaling \modelname \ to proteins larger than 128 proved to be challenging and remains to be a direction of future work.

\section{Additional metric details}\label{sec:additional_metrics}
\textbf{Self-consistency algorithm.}
\Cref{sec:sc} described our self-consistency metrics for evaluating the designability of backbones generated with \modelname{}.
\Cref{alg:self_consistency} makes explicit the procedure we use for computing these metrics.

% \begin{wrapfigure}{R}{0.5\textwidth}
\begin{figure}[H]
\begin{minipage}{0.5\textwidth}
  \begin{algorithm}[H]
    \caption{Self-consistency calculation}
    \label{alg:self_consistency}
    \begin{algorithmic}[1]
        \Require $\X \in \mathbb{R}^{N, 3}$ 
        \For{$i \in 1,\dots,8$}
            \State $\Seqval_i \gets$ \texttt{ProteinMPNN}($\X$)
            \State $\hat{\X}_i \gets$ \texttt{AF2}($\Seqval_i$)
        \EndFor
        % \State \texttt{sc\_rmsd} $\gets \min_{i \in 1,\dots,8}$\texttt{RMSD}$(\hat{\X}_i, \X)$
        \State \texttt{sc\_tm} $\gets \max_{i \in 1,\dots,8}$\texttt{TMscore}$(\hat{\X}_i, \X)$
        % \Ensure \texttt{sc\_rmsd}, \texttt{sc\_tm}
        \Ensure, \texttt{sc\_tm}
    \end{algorithmic}
  \end{algorithm}
\end{minipage}
% \end{wrapfigure}
\end{figure}

\paragraph{Using dihedral angles to calculate helix chirality.}
Natural proteins are chiral molecules that contain only right-handed alpha helices. 
However, because the underlying EGNN in our model is equivariant to reflection, it can produce samples with left-handed helices. 
While examining model samples, we additionally observed samples with both left and right-handed helices (\Cref{fig:left_handed_examples}), even though in theory the EGNN should be able to detect and avoid the chiral mismatch. 
Left-handed helices are fundamentally invalid geometries in proteins and represent a trivial failure mode when calculating the self-consistency and other metrics.
Samples with a mixture of left and right-handed helices are especially problematic because they cannot be corrected simply by reflecting the coordinates.
As such, it is important to identify and separate samples with mixed chirality.

To detect chiralty, we compute the dihedral angle between four consecutive C-$\alpha$ atoms as a chiral metric to distinguish between the two helix chiralities.
Algorithmically, for every C-$\alpha$ \texttt{i}, we calculate the dihedral between  C-$\alpha$ \texttt{i}, \texttt{i+1}, \texttt{i+2}, and \texttt{i+3}.
C-$\alpha$ \texttt{i} with dihedral angles between 0.6 and 1.2 radians are classified as right-handed helices, and angles between -1.2 and -0.6 are classified as left-handed helices, with everything else classified as non-helical. Because C-$\alpha$ atoms in native helices tend to form contiguous stretches longer than one residue in the primary sequence, helical stretches less than one amino acid were removed. This filtering is meant to help avoid accidentally counting the occasional isolated backbone geometry that falls into a helical bin as a true helix. 
Finally, for all C-$\alpha$ atoms \texttt{i}  that are still categorized as part of a helix, the associated \texttt{i+1}, \texttt{i+2} and \texttt{i+3} C-$\alpha$ atoms are also counted as part of that helix.

\section{Additional experimental results}\label{sec:additional_results}

In this section, we describe additional results to complement the main text.
We provide a description of the motif targets in \Cref{sec:conditional_sampling}, along with results of a scaffolding failure case in \Cref{supp:additional_motif_scaffold}.
To understand the qualitative outcomes of \scTM, we present additional results of backbone designs, their AF2 prediction, and most closely related PDB parent chain for different thresholds of \scTM{} in \Cref{supp:sctm_qualitative}.
We provide additional examples of latent interpolations in \Cref{supp:latent}.
Finally, \Cref{supp:clustering} presents a structural clustering of unconditional backbone samples; this result provides further evidence of \modelname{}'s ability to generate diverse backbone structures.

\begin{figure}[H]
    \centering
    \includegraphics[width=0.9\textwidth]{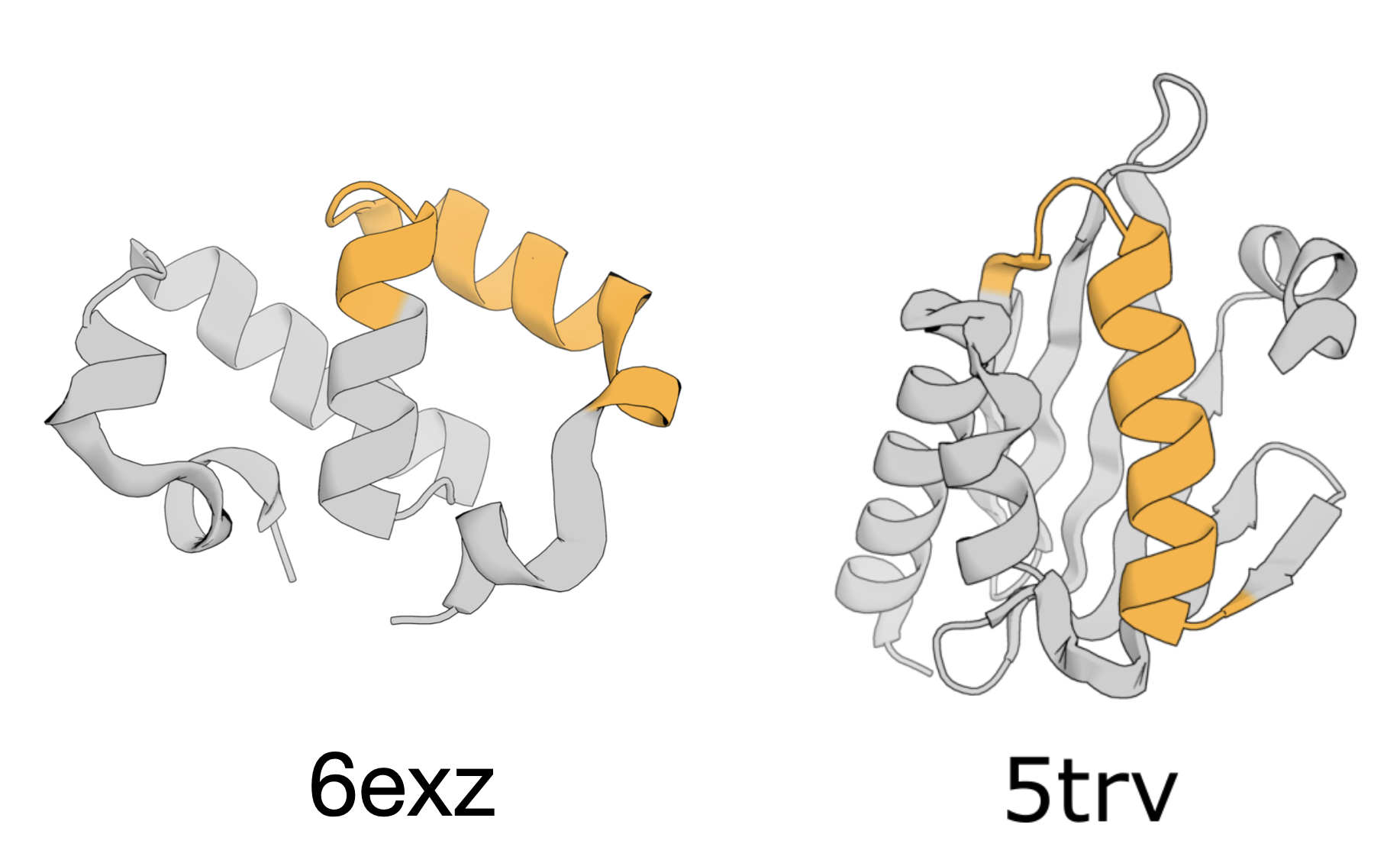}
    \caption{Structures used for motif-scaffolding test cases. Native structures (grey) and their motifs (orange) that were used for the motif-scaffolding task are shown.}
    \label{fig:wt_motifs}
\end{figure}

\subsection{Additional motif-scaffolding results}
\label{supp:additional_motif_scaffold}
We here provide additional details of the motif-scaffolding experiments described in \Cref{sec:experiments}.
\Cref{table:inpainting_targets} specifies the total lengths, motif sizes, and motif indices of our test cases.
In \Cref{fig:wt_motifs} we depict the structures of the native proteins (6exz and 5trv) from which the motifs examined quantitatively in the main text were extracted.
\Cref{fig:failure_modes} analyzes commonly observed failure modes of \modelname{} backbone samples involving chain breaks, steric clashes, and incorrect chirality.

\Cref{fig:rsv_failure} presents quantitative results on a harder inpainting target.
% In this case, the motif is obtained from the respiratory syncytial virus (RSV) protein \texttt{5tpn},
% and the motif is residues 163-181 of chain A of of native structure.
In this case, the motif is defined as residues 163--181 of chain A of respiratory syncytial virus (RSV) protein (PDB ID: 5tpn).
We attempted to scaffold this motif into a 62 residue protein, with the motif as residues 42--62.
We chose this placement because previous work \citep{wang2021deep} identified a promising candidate scaffold with this motif placement.
In contrast to the cases described in the main text,
for which a suitable scaffold exists in the training set,
\samplingname{} and the other inpainting methods failed to identify scaffolds that recapitulated this motif to within a motif RMSD of 1 \AA.

\begin{table}[H]
\centering
    \caption{Motif-scaffolding test case additional details.}
    \label{table:inpainting_targets}
    \begin{tabular}{ |p{3.0cm}||p{2cm}|p{4cm}|  }
     \hline
     Origin/ Protein & Total length & Motif size (residue range)\\
     \hline
     \hline
    6exz &   72  & 15 (30--44) \\
     \hline
    % `6exz` &   69  & 10 (25-35) \\
    %  \hline
    % `5ci9` &   116  & 20 (60-80) \\
    %  \hline
    5trv &   118  & 21 (42--62) \\
     \hline
    \texttt{RSV} (PDB-ID: 5tpn) &  62  & 19 (16--34) \\
     \hline
    \texttt{EF-hand} (PDB-ID: 1PRW) &  53  & 5 (0--4), 13 (31--43) \\
     \hline
    \end{tabular}
\end{table}

\begin{figure}[H]
    \centering
    \includegraphics[width=0.9\textwidth]{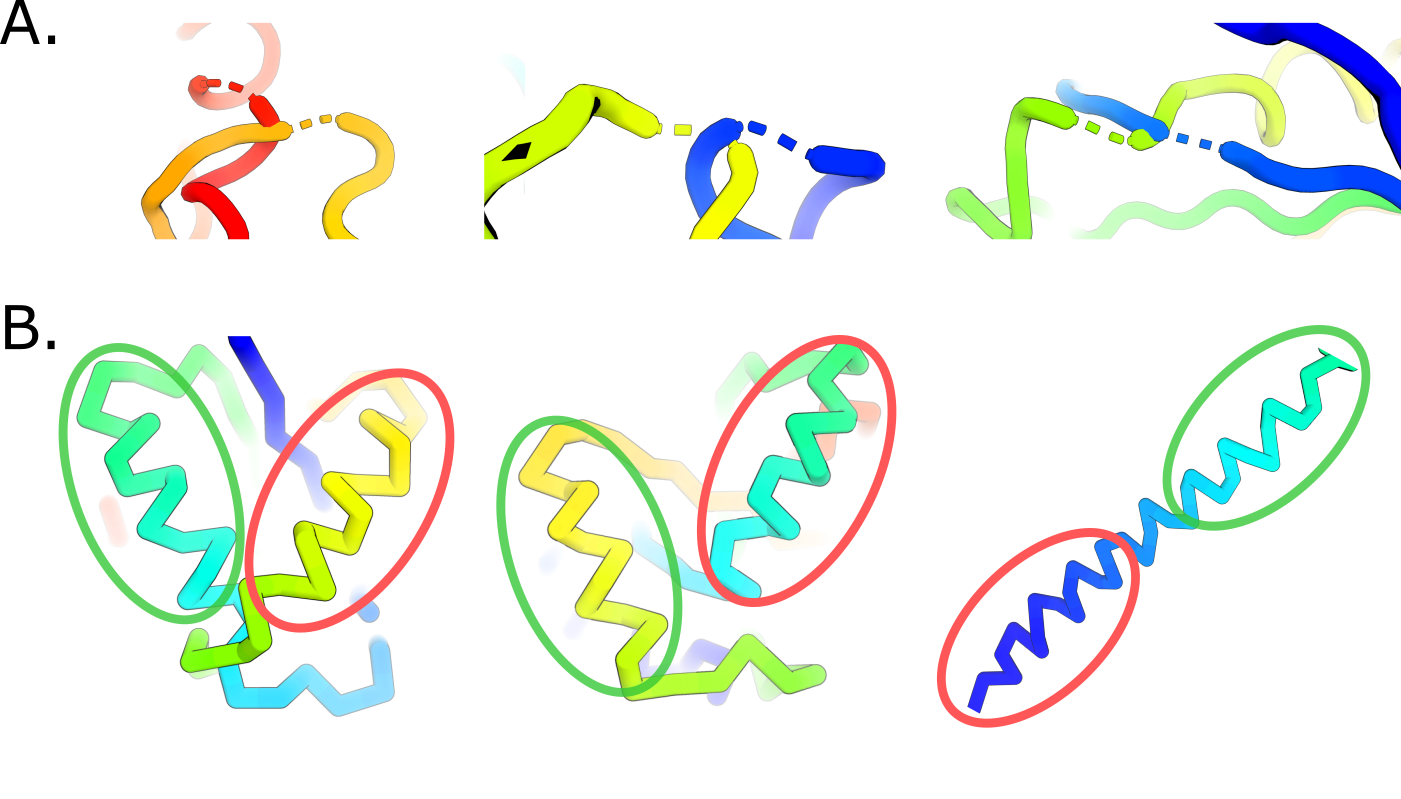}
    \caption{Failure modes in \modelname{} backbone samples. (A) Backbone clashes and chain breaks. The C-$\alpha$ atoms can be spaced further than the typical 3.8\AA{} between neighbors, resulting in a chain break (dashed lines). Additionally, backbone segments can be too close to each other, resulting in obvious overlaps and clashes. (B) Backbones with a mixture of left (circled in red) and right (circled in green) handed helices. These chirality errors cannot be corrected simply by mirroring the sampled backbone.}
    \label{fig:failure_modes}
\end{figure}

\begin{figure}[H]
    \centering
    \includegraphics[width=0.9\textwidth]{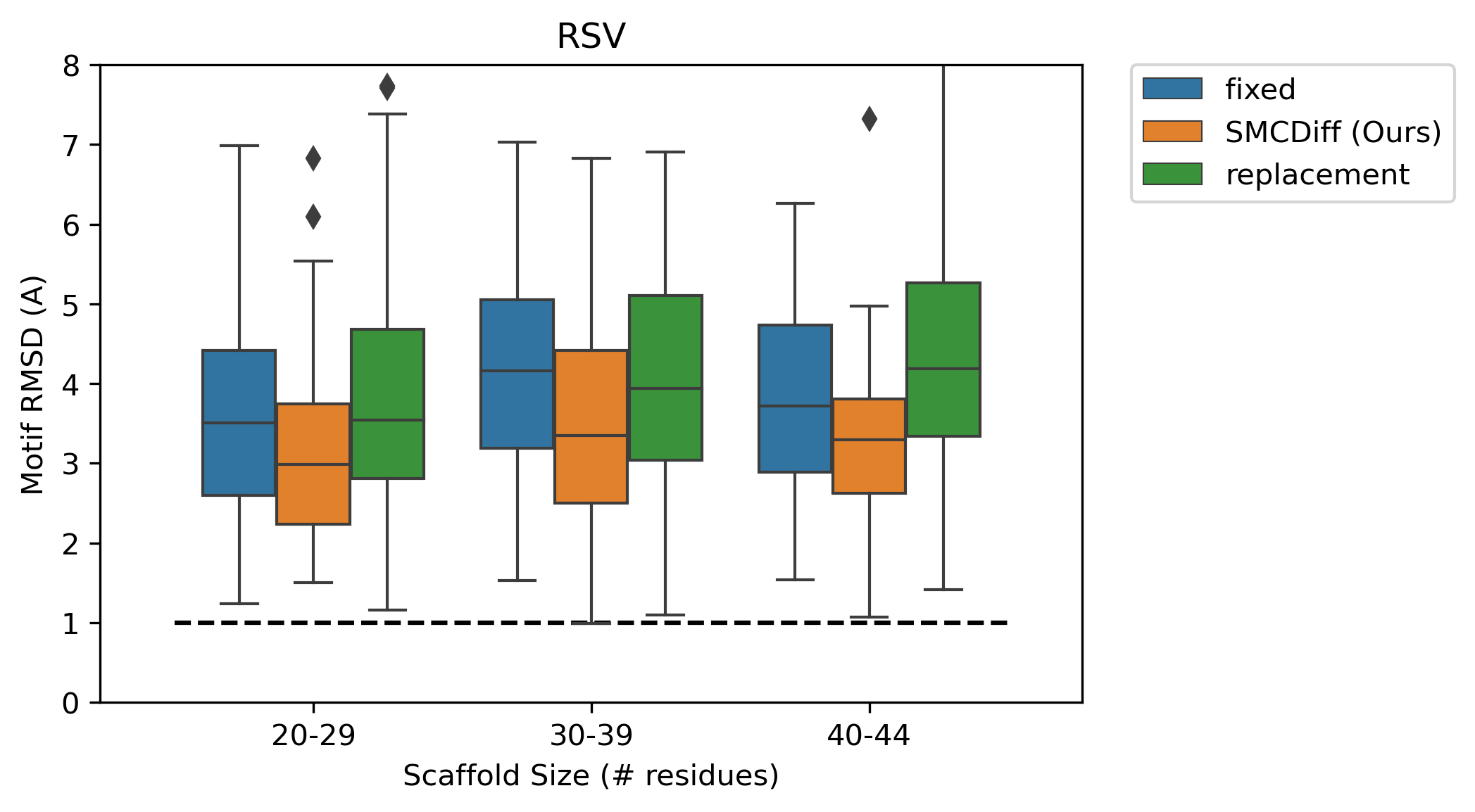}
    \includegraphics[width=0.9\textwidth]{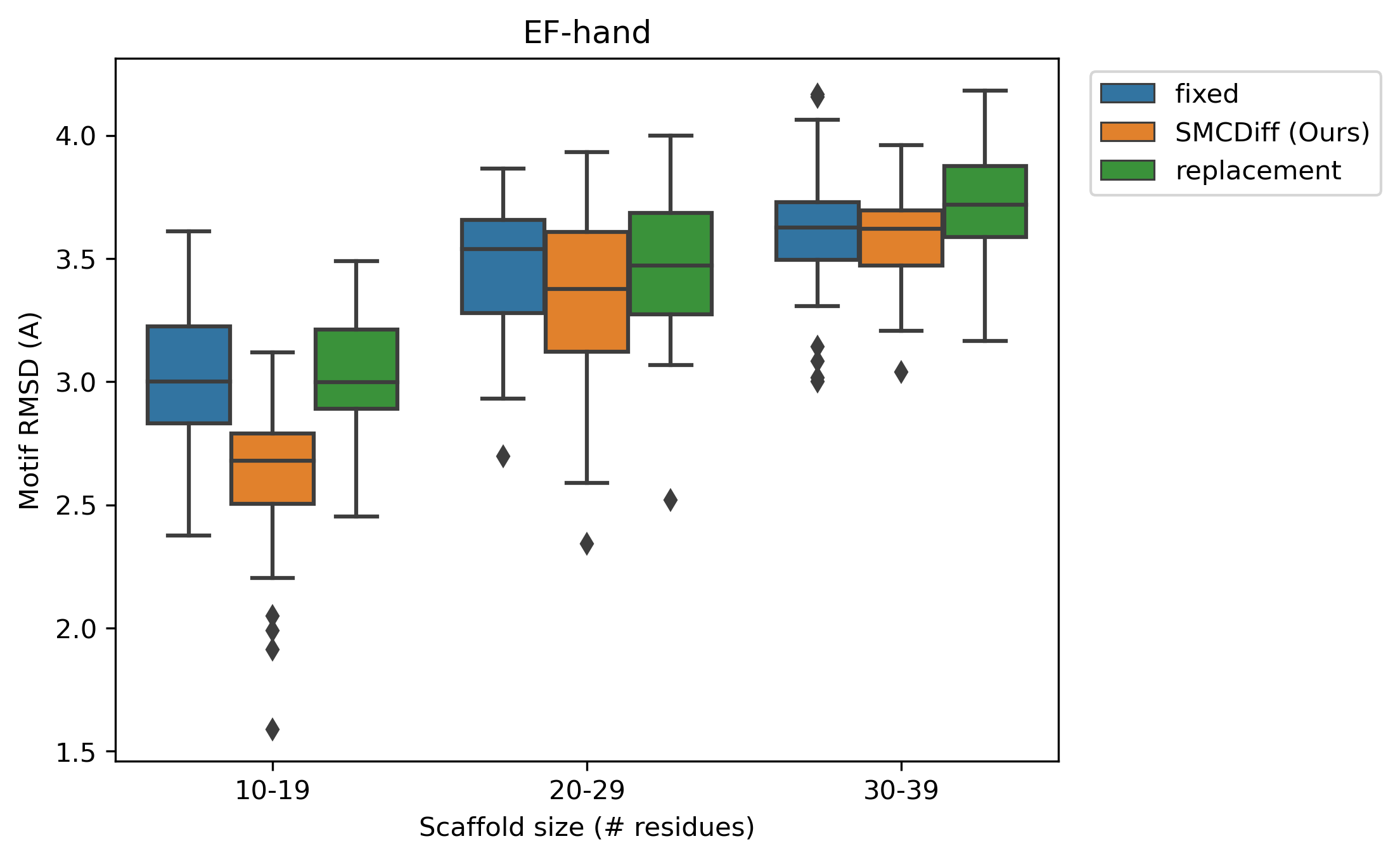}
    \caption{Additional inpainting results on a more challenging motif extracted from the respiratory syncytial virus (RSV) and EH-hand motif.
    The three inpainting methods are evaluated as described in \Cref{sec:experiments}.
    }
    \label{fig:rsv_failure}
\end{figure}

\subsection{Qualitative analysis of \scTM{} in different ranges}
\label{supp:sctm_qualitative}

In this section, we give intuition for backbone designs and AF2 predictions associated with different values of \scTM{} to aid the interpretation of the \scTM{} results provided in \Cref{sec:experiments}.
\Cref{fig:sctm_qualitative} examines a possible categorization of \scTM{} in three ranges. The first two rows correspond to backbone designs that achieve \scTM{} > 0.9. We see the backbone designs in the first column closely match the AF2 prediction in the second column. A closely related PDB example can be found when doing a similarity search of the highest PDB chain with the highest TM-score to the AF2 prediction.
We showed in \Cref{fig:unconditional}B that \scTM{} > 0.9 is indicative of a close structural match being found in PDB.

The middle two rows correspond to designs that achieve \scTM{} $\sim$ 0.5. These are examples of backbone designs on the edge of what we deemed as designable (\scTM{} > 0.5). In these cases, the AF2 prediction shares the same coarse shape as the backbone design but possibly with different secondary-structure ordering and composition. In the length 69 example, we see the closest PDB chain has a TM-score of only 0.65 to the AF2 prediction but roughly the same secondary-structure ordering as the backbone design.
The length 100 sample is a similar case of AF2 producing a roughly similar shape to the backbone design, but has no matching monomer in PDB.
% Designs in this category often do not have a closely matching chain in PDB, but the variability of the backbone designs and AF2 predictions does not allow for concluding that \modelname{} can produce novel folds.

The final category of \scTM{} < 0.25 reflects failure cases when \scTM{} is low.
% The sequences generated by \protmpnn{} in these cases are generally poor. For instance the sequence of the length 54 sample is \TODO{Insert sequence}.
The AF2 predictions in this case have many disordered regions and bear little structural similarity with the original backbone design.
Similar PDB chains are not found.
We expect that improved generative models of protein backbones would not produce any samples in this category.

\begin{figure}[H]
\centering
\includegraphics[width=0.85\textwidth]{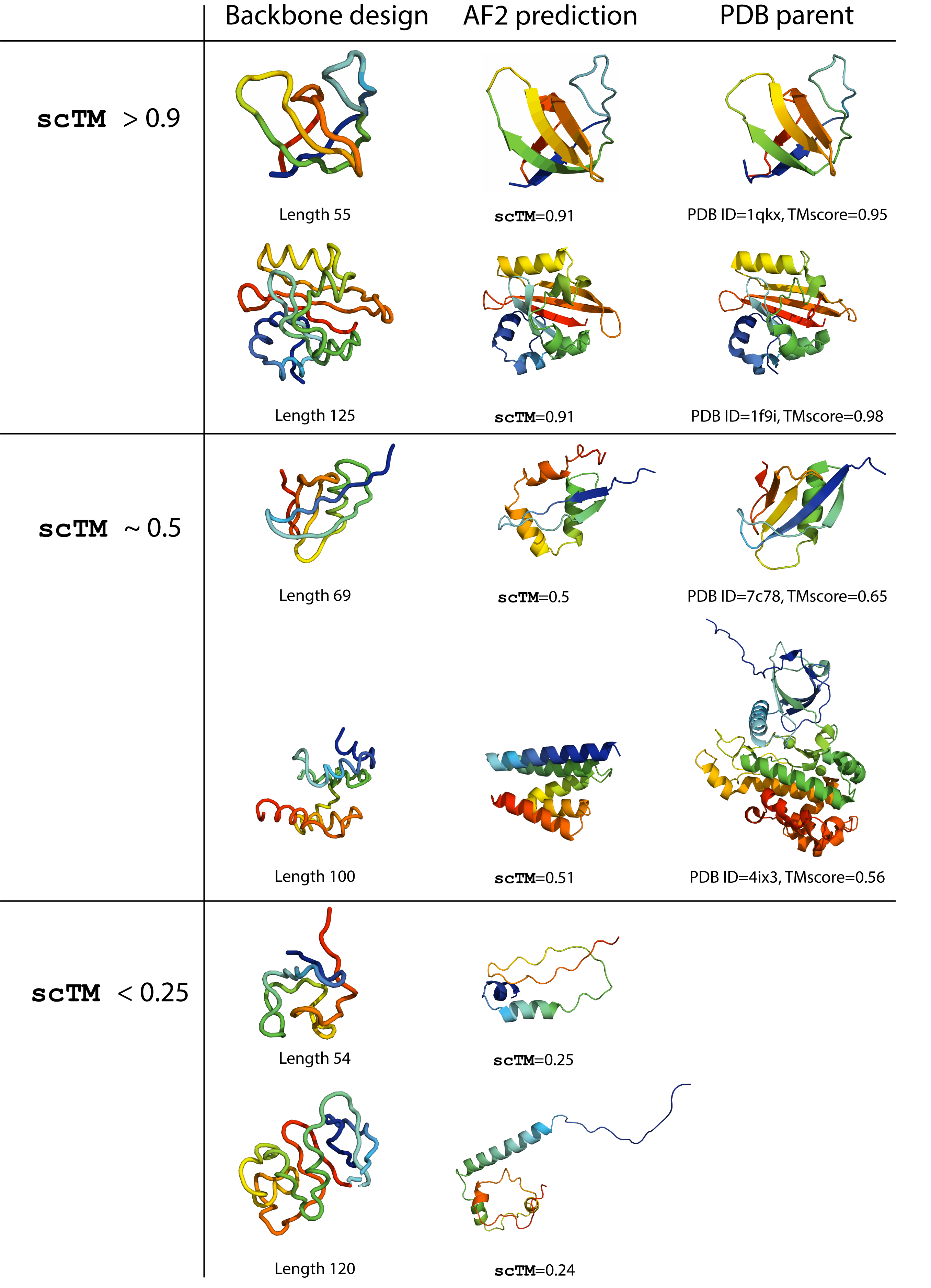}\
    \caption{Qualitative analysis of unconditional backbone samples from \modelname. The first column displays backbone designs from \modelname{} and their sequence lengths.
    The second column displays the highest \scTM{} scoring AF2 predictions from the \protmpnn{} sequences of the corresponding backbone design in the first column. The third column displays the closest PDB chain to the AF2 prediction in the second column with the PDB ID and TM-score written below. The third column is blank for the last two rows since no PDB match could be found.
    See \Cref{supp:sctm_qualitative} for discussion.}
    \label{fig:sctm_qualitative}
\end{figure}

\subsection{Additional latent interpolation results}\label{sec:additional_interpolations}
\label{supp:latent}
We here provide additional latent interpolations.
\Cref{fig:latent_interpolation_len_89,fig:latent_interpolation_len_63} depict interpolations for between model samples for lengths 89 and 63, respectively.

\begin{figure}[H]
    \centering
   \includegraphics[width=0.9\textwidth]{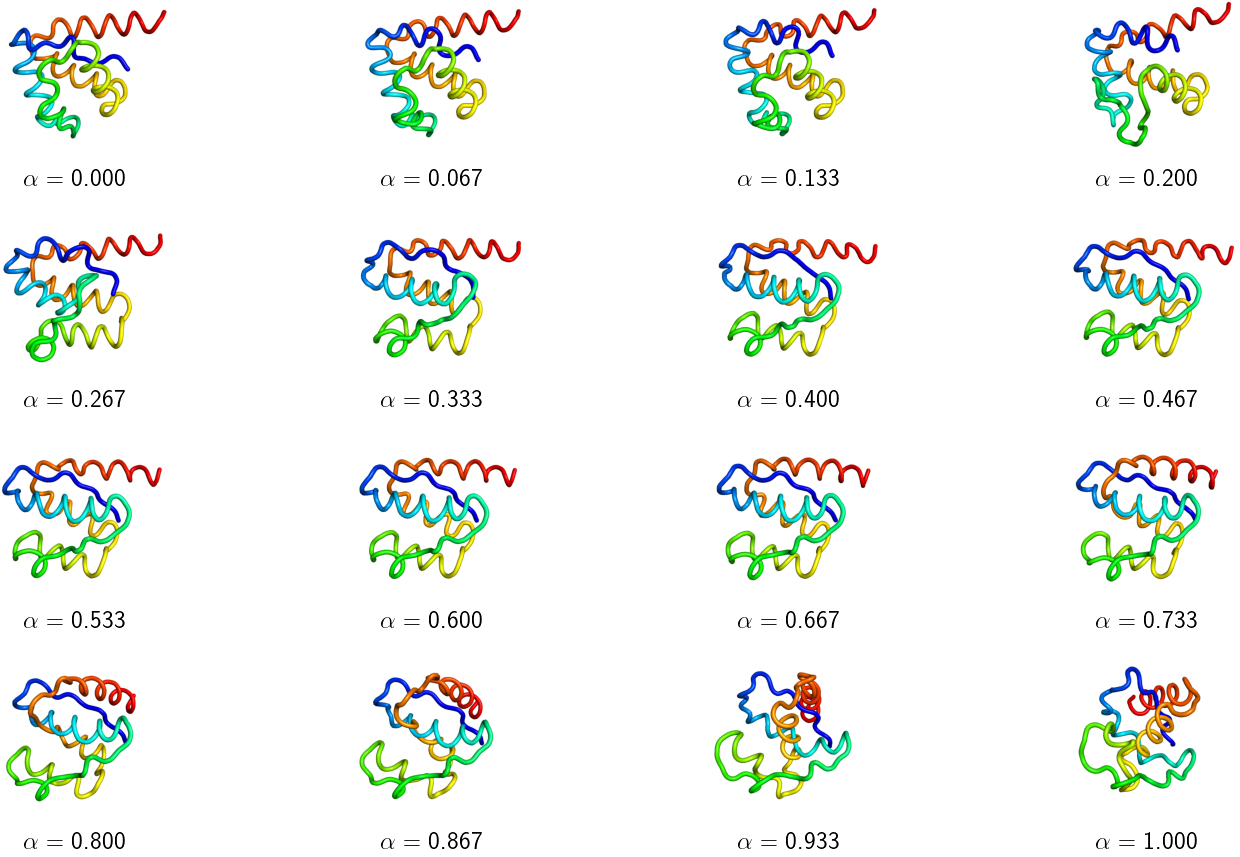}
    \caption{Latent interpolation of length 89 backbone sample from $\alpha=0$ to 1.}
    \label{fig:latent_interpolation_len_89}
\end{figure}

\begin{figure}[H]
    \centering
    \includegraphics[width=0.9\textwidth]{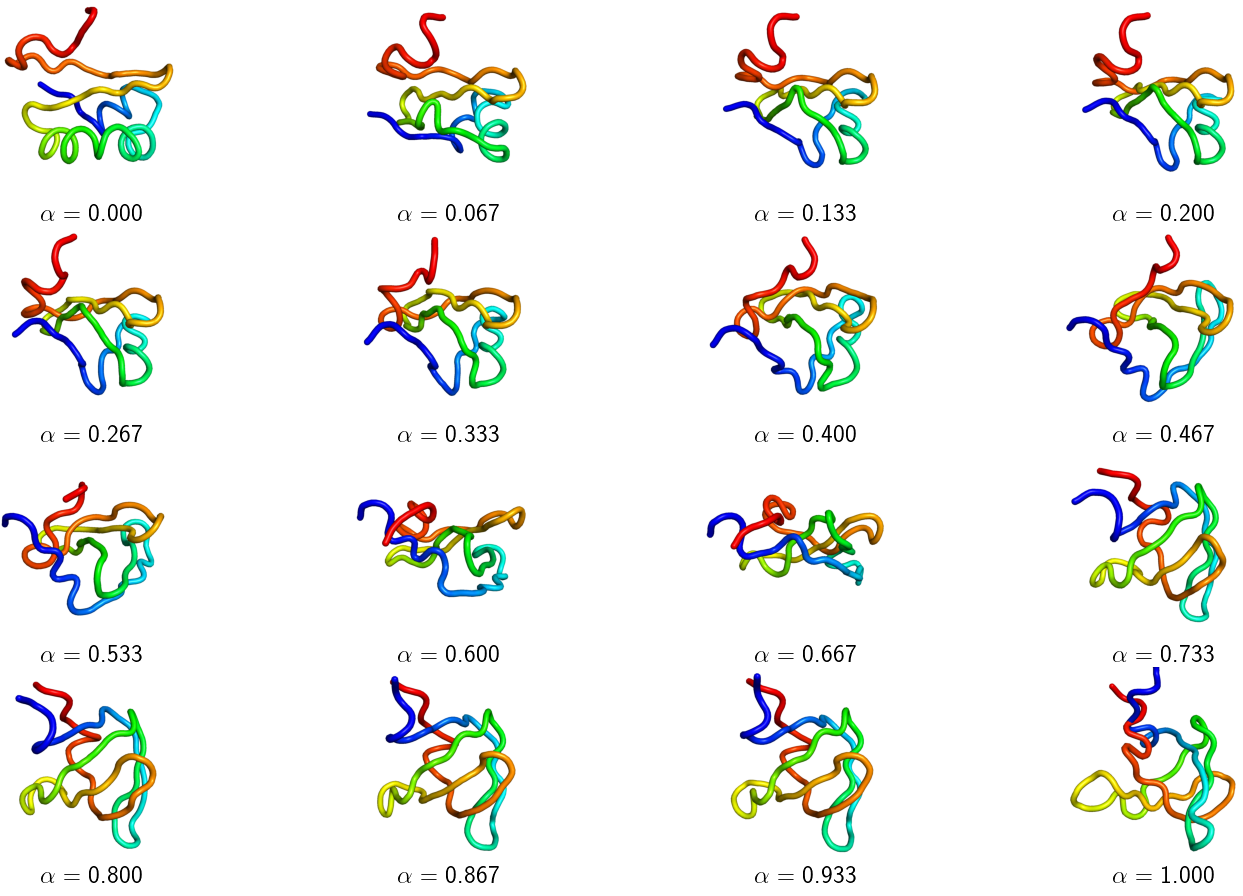}
    \caption{Latent interpolation of length 63 backbone sample from $\alpha=0$ to 1.}
    \label{fig:latent_interpolation_len_63}
\end{figure}

\newpage

\subsection{Structural clustering}
\label{supp:clustering}

All 92 samples with \scTM{} > 0.5 were compared and clustered using MaxCluster \cite{herbertmaxcluster}. Structures were compared in a sequence independent manner, using the TM-score of the maximal subset of paired residues. They were subsequently clustered using hierarchical clustering with average linkage, 1 - TM-score as the distance metric and a TM-score threshold of 0.5 (\Cref{fig:clustering} A).

\begin{figure}[H]
    \centering
    \includegraphics[width=1.0\textwidth]{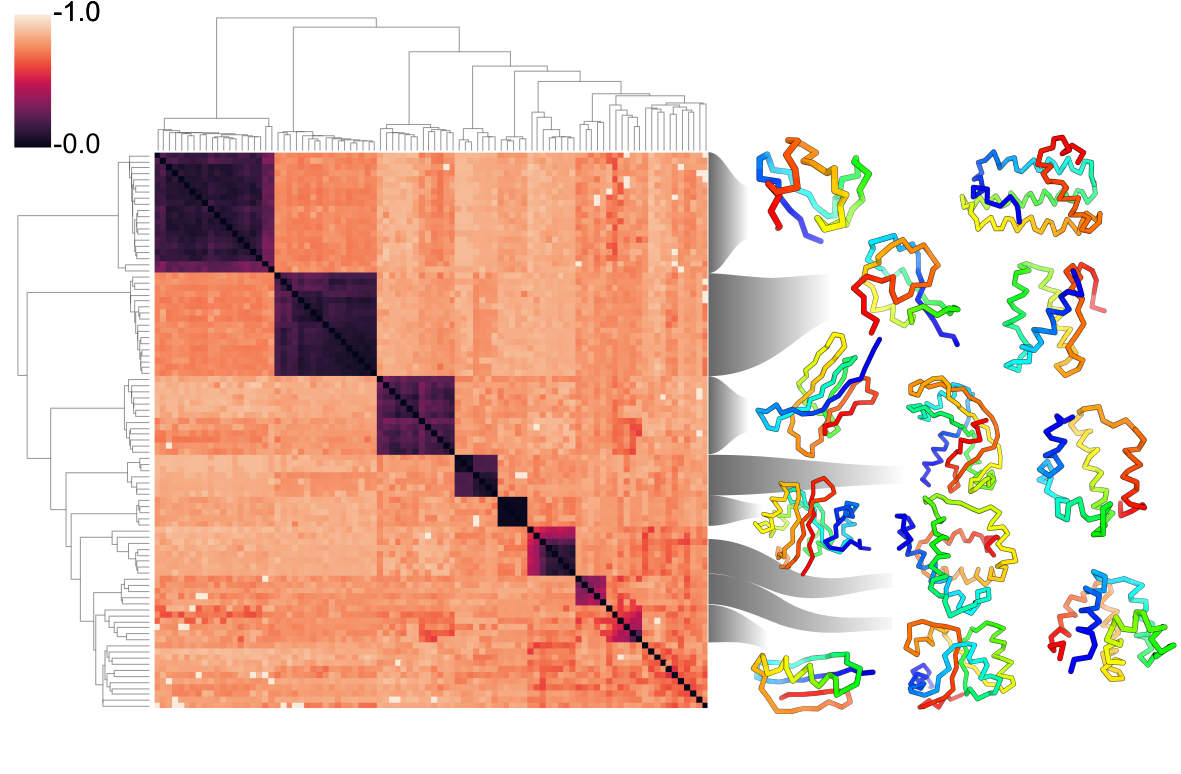}
    \caption{Clustering of self-consistent \modelname{} samples. The distance matrix is 1 - TM-score between pairs of samples, and ranges from 0 (exact matach) to 1 (no match). Dendrograms are from hierarchical clustering using the average distance metric. Designs on the right are cluster centroids. Gray lines connect larger clusters with more than one member to its centroid, while the remaining designs are from a random selection of the remaining single-sample clusters. Protein backbones are colored from blue at the N-terminus to red at the C-terminus.}
    \label{fig:clustering}
\end{figure}

\section{Applicability of \samplingname{} beyond proteins: MNIST inpainting}\label{sec:mnist_inpaint}

Our goal in this section is to study the applicability of \samplingname{} beyond motif-scaffolding,
by applying it to inpainting on the MNIST digits dataset.
We compare \samplingname{} with the replacement method on the task of sampling the remaining half of MNIST digits.
We first train DDPM with $\beta_1 = 10^{-4}, \beta_T=0.2, T=1000$ using a small 8-layer CNN on MNIST with batch size 128 and ADAM optimizer for 100 epochs until it is able to generate reasonable MNIST samples (\Cref{fig:unconditional_mnist}).
We then selected 3 random MNIST images and occluded the right half.
The left half would then serve as the conditioning information to the diffusion model (\Cref{fig:mnist_conditioning}).

\begin{figure}[H]
    \begin{minipage}[c]{0.35\linewidth}
        \includegraphics[width=\linewidth]{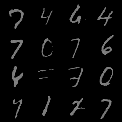}
        \caption{Unconditional MNIST samples.}
        \label{fig:unconditional_mnist}
    \end{minipage}
    \hfill
    \begin{minipage}[c]{0.5\linewidth}
        \includegraphics[width=\linewidth]{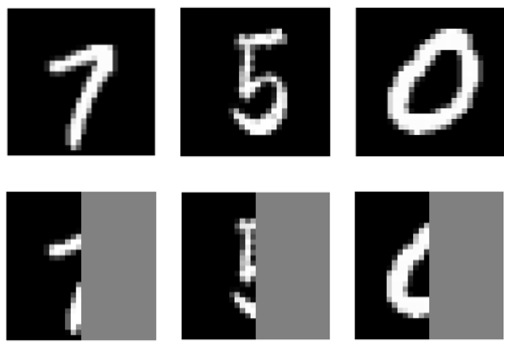}
        \caption{Full MNIST images and their occluded halves used for inpainting experiments.}
        \label{fig:mnist_conditioning}
    \end{minipage}
\end{figure}
\vspace{-10pt}

For each occluded image, we fixed a single forward trajectory and sampled 16 images from each method: replacement method and \samplingname{} with 16 or 64 particles $(K)$. Results are shown in \cref{fig:inpaint_mnist}. We observe the replacement method can sometimes produce coherent samples as a continuation of the conditioning information, but more often it attempts to produce incoherent digits. \samplingname{} on the other hand tends to produce digits that compliment the conditioning information.
For more difficult occlusions, such as 5 and 0, \samplingname{} can still fail although increasing the number of particles ($K=64$) tends to produce samples that are more visually coherent.

It is important to note SMCDiff has additional computation overhead based on the number of particles. It can be more expensive than replacement method but result in higher quality samples. 
Investigating \samplingname{} in more difficult datasets with improved architectures is a direction of future research.

\begin{figure}[H]
    \centering
    \includegraphics[width=0.9\textwidth]{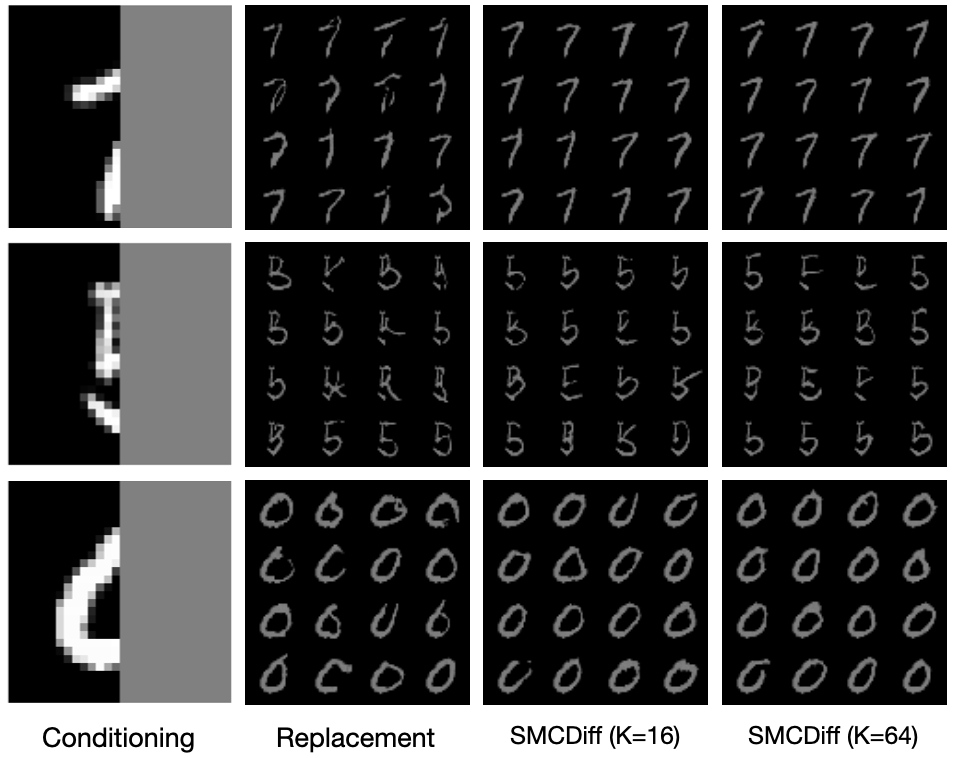}
    \caption{MNIST inpainting results for replacement and \samplingname{}. See text for explanation.}
    \label{fig:inpaint_mnist}
\end{figure}

% \nopagebreak
% \begin{figure}[H]
% \vskip 0.2in
% \hfill
% \subfigure[scTM success rate. \TODO{Success rate with sequence diffusion.}]{\includegraphics[width=0.3\textwidth]{figures/unconditional_sctm.png}}
% \hfill
% \subfigure[Clustering of samples by TM-score. \TODO{Analyze cluster centroids.}]{\includegraphics[width=0.3\textwidth]{figures/unconditional_clusters.png}}
% \hfill
% \subfigure[Similarity to PDB.]{\includegraphics[width=0.3\textwidth]{figures/unconditional_training_set_comparison.png}}
% \hfill
% \caption{s}
% \end{figure}

\end{document}